%% file: main.tex
\documentclass[runningheads]{llncs}
\usepackage[T1]{fontenc}
\usepackage{graphicx}
\usepackage{hyperref}
\usepackage{color}

\usepackage[utf8]{inputenc}
\usepackage[english]{babel}
\usepackage{mathrsfs,amsmath}
\usepackage{listings}
\usepackage{csquotes}
\usepackage{xspace}
\usepackage{selectp}
\usepackage{soul}
\usepackage{subcaption}
\usepackage{mathptmx}
\usepackage{todonotes}
\usepackage{url}
\usepackage{yfonts}
\usepackage{makecell}
\usepackage[shortlabels]{enumitem}
\usepackage{amsfonts}
\usepackage{amssymb}
\usepackage{verbatim}
\usepackage{epigraph}

\usepackage{tikz}

\usepackage{comment}

\begin{document}

\include{macros}

\title{SM-based Semantics for Answer Set Programs Containing Conditional Literals and Arithmetic}
\titlerunning{Semantics for Programs Containing Conditional Literals and Arithmetic}
%
\author{Zachary Hansen\orcidID{0000-0002-8447-4048} \and Yuliya Lierler\orcidID{0000-0002-6146-623X}}
\authorrunning{Z. Hansen and Y. Lierler}
%
\institute{University of Nebraska Omaha, Omaha NE 68106, USA}
\maketitle              
\begin{abstract}
Modern answer set programming solvers such as \clingo\ support advanced language constructs that improve the expressivity and conciseness of logic programs.
Conditional literals are one such construct. They form ``subformulas'' that behave as nested implications within the bodies of logic rules. Their inclusion brings the form of rules closer to the less restrictive syntax of first-order logic.
These qualities make conditional literals useful tools for knowledge representation.
In this paper, we propose a semantics for logic programs with conditional literals and arithmetic based on the $\SM$ operator. 
These semantics do not require grounding, unlike the established semantics for such programs that relies on a translation to infinitary propositional logic. 
The main result of this paper establishes the precise correspondence between the proposed and existing semantics. 

\end{abstract}
%
%
%


\section{Introduction}

Answer Set Programming (ASP) \cite{Marek1999,Niemela1999} is a declarative programming paradigm
that has been applied within a variety of challenging and high-consequence systems such as explainable donor-patient matching~\cite{cab21}, space shuttle decision support systems~\cite{bal05,bal01}, train scheduling~\cite{abe21}, robotics~\cite{geb18asprilo}, and automated fault diagnosis~\cite{wotkau22}.
ASP programs are concise, human-readable, and benefit from well-defined semantics rooted in mathematical logic -- these qualities make ASP programs attractive candidates for formal verification~\cite{cafali20a}.
Providing high levels of assurance regarding program behavior is particularly crucial for safety-critical applications.
%
%
This paper is part of a research stream with the long-term goal of supporting rigorous verification of ASP systems.

Conditional literals are powerful language features for knowledge representation.
Originating in the \textsc{lparse} grounder~\cite{syr04}, these languages features are now also supported by the answer set solver \clingo~\cite{potasscoManual,geb15}.
Intuitively, they represent a nested implication within the body of an ASP rule~\cite{hanlie22}.
Rules with conditional literals concisely express knowledge that may be difficult to otherwise encode.
For instance, conditional literals are widely employed in meta-programming --
Listings 4-7 in ``How to build your own ASP-based system?!" by Kaminski et al.~\cite{karoscwa20a} define meta encodings which compute the classical and supported models of reified logic programs; these encodings rely heavily on conditional literals.

Conditional literals may also make programs easier to formally verify by reducing the number of auxiliary or inessential predicates in a program.
Consider Listing~\ref{list:gc_yl}, which contains a typical encoding of the graph coloring  problem.
\begin{lstlisting}[
  caption = {Graph coloring problem encoding. 
},
  label={list:gc_yl},
  basicstyle=\ttfamily\small,
  numbers=left,
  stepnumber=1,
]
{asg(V, C)} :- vtx(V), col(C).
:- asg(V, C1), asg(V, C2), C1 != C2.
colored(V) :- asg(V,C).
:- vtx(V), not colored(V).
:- asg(V1, C), asg(V2, C), edge(V1, V2).
\end{lstlisting}
The program in this listing can be simplified by replacing lines 3-4 with the following \emph{constraint} (a rule with an empty head) containing a conditional literal:
\begin{gather}
    \label{eq:rule.color.cl.constraint}
    \ruleo ~not ~asg(V, ~C) ~: ~col(C); ~vtx(V).
\end{gather}
This simplification is attractive  since it eliminates the auxiliary predicate $colored/1$ (introduced in line~3 for the sole  purpose of stating the subsequent constraint in line~4).
We will use rule~\eqref{eq:rule.color.cl.constraint} as a running example in the remainder of the paper.
%

Typically, the semantics of programs with variables are  defined indirectly, via a procedure called \emph{grounding}.
Grounding turns a given program with variables into a propositional one. Then, semantics are defined for the resulting propositional program.
This hampers our ability to reason about the behavior of programs independently of a specific grounding context.
%
In 2011, Ferraris, Lee, and Lifschitz proposed semantics for answer set programs that bypasses grounding~\cite{ferleelif11a}.
They introduced the $\SM$ operator, which turns a program (or, rather, the first-order logic formula associated with the considered program) into a classical second-order formula.
The Herbrand models of this formula coincide with the answer sets of the original program. 

Since then, that approach  has been generalized  to cover such features of ASP input languages  as  aggregates~\cite{fan22,fanhanlie24},
arithmetic~\cite{falilusc20a,lif19},  and  conditional literals~\cite{hanlie22}. 
Yet, all of these features were addressed independently of the others. This paper helps to close that gap.
Here, we introduce grounding-free $\SM$ operator-based  semantics for logic programs containing {\em both} conditional literals and arithmetic, combining ideas from earlier work on these features~\cite{falilusc20a,hanlie22,lif19}. We also show that the proposed characterization coincides with the existing  semantics for such programs based on grounding to infinitary propositional logic~\cite{geb15}; 
these are the semantics adhered to by \clingo.

%

One of the advantages of the $\SM$-based characterization is that it enables us to construct proofs of correctness in a modular way that does not rely on grounding the program with respect to a specific instance of input data.
For instance, this style of verification -- which exploits the $\SM$ operator's ability to divide programs into modules -- has been employed to demonstrate the adherence of Graph Coloring, Hamiltonian Cycle, and Traveling Salesman problems to natural language specifications~\cite{cafali20a,fanhanlie22b}.
This paper extends the class of programs for which such arguments can be constructed.
Section~\ref{sec:syntax} defines the language of logic programs considered, and Section~\ref{sec:semantics:sm}  provides the essence of 
 the $\SM$ characterization for logic programs with conditional literals and arithmetic.
Section~\ref{sec:semantics:ipl} reviews the established semantics for programs with conditional literals and arithmetic, which relies on a translation from logic programs to infinitary propositional formulas.
The main results of this paper --Theorems~\ref{thrm:connecting.tsm} and \ref{thrm:main} -- are given in Section~\ref{sec:consem}, connecting our $\SM$-based semantics to the established semantics.

\input{content}




\section{Conclusions and Future Work}


In this paper we introduced semantics based on the $\SM$ \nobreakdash operator for logic programs containing both conditional literals and arithmetic.
The key result of this work -- Theorem~\ref{thrm:main} -- demonstrates that the definition of answer sets using our extension of the $\tau^*$ translation correctly characterizes the behavior of the answer set solver \clingo.

%
%
%
The main intuition of the $\phi$ translation~\cite{hanlie22}, which provided $\SM$-based semantics for programs with conditional literals but \emph{without} arithmetic, was that conditional literals in rule bodies behave like nested implications. 
Our proposed translation preserves this intuition. 
However, allowing for arithmetic in the language considered here made the argument of the correspondence between our definition of answer sets and gringo answer sets substantially more complex in comparison to the similar argument
made for the case of translation~$\phi$.
Specifically, the correspondence now relies on strong equivalence as opposed to syntactic identity.

The newly introduced $\SM$-based semantics enables ASP practitioners to verify programs with conditional literals and arithmetic in the style of past work on modular verification~\cite{cafali20a,fanhanlie22b}.
The definition of $\bold{p}$-answer sets offers greater flexibility (due to the possibility to distinguish intensional and extensional predicate symbols) than the traditional notion of gringo answer sets and supports this verification style.
Furthermore, conditional literals can make programs more concise and easier to verify.
For example, as illustrated in the Introduction, we can refactor Listing~\ref{list:gc_yl} to use a smaller set of  predicate symbols by employing a conditional literal in constraint~\eqref{eq:rule.color.cl.constraint}.
%
%

In future work we intend to investigate how the process of using conditional literals to eliminate auxiliary predicates can be generalized.
Automated verification is another important direction for future work.
Recent progress in this direction is manifested in the \anthem\footnote{\url{https://github.com/potassco/anthem}} system, which employs an automated theorem prover to establish strong~\cite{lifpeaval01} or external~\cite{fanhanliftem23} equivalence of \mg\ programs.
We plan to extend the theory and implementation supporting both types of verification to the language presented in this paper.
%
Similar verification tools include \textsc{ccT}~\cite{oetseitomwol09} and \textsc{lpeq}~\cite{bomjannie20,janoik04}.
Future work will include detailed comparisons of such systems against \anthem.

Finally, it is worth investigating ways to simplify the formulas produced by our extension of $\tau^*$.
For instance, it is easy to see that the formulas in Example 1 contain several unnecessary
existential quantifiers.
By applying ht-equivalent simplifications, we obtain a considerably more readable translation:
$$
\forall V \left( \forall C (col(C) \to \neg asg(V,C)) \wedge vtx(V) \to \bot\right).
$$
Past work in this vein has been devoted to developing a ``natural" translation $\nu$ for a
broad fragment of the \mg\ language~\cite{lif21nu}.
%
%
We plan to extend the $\nu$ translation with conditional literals to make (automated) verification of programs easier.

\begin{credits}
\subsubsection{\ackname} We are grateful to 
Jorge Fandinno and Vladimir Lifschitz
for their valuable comments,
and to our anonymous reviewers for their feedback.

\subsubsection{\discintname} The authors have no competing interests. 
\end{credits}

%
%
%
\bibliographystyle{splncs04}
\bibliography{bib}

\appendix

\section{Proofs of Theoretical Claims of the Paper (Propositions~1-4)}
\input{appendix-1}

\input{appendix-2}
\input{appendix-3}

\input{appendix-4}

\end{document}

%% file: macros.tex
\newtheorem{lem}{Lemma}
\newtheorem{thm}{Theorem}
\newtheorem{cor}{Corollary}
\newtheorem{fact}{Fact}

\def\ar{\leftarrow}
\def\lrar{\leftrightarrow}
\def\beq{\begin{equation}}
\def\eeq#1{\label{#1}\end{equation}}
\def\ba{\begin{array}}
\def\ea{\end{array}}
\def\gringo{{\sc gringo}}
\def\clingo{{\sc clingo}}
\def\anthem{{\sc anthem}}
\def\vampire{{\sc vampire}}
\def\num{\overline}
\def\p2f{\hbox{p2f}}
\def\no{\emph{not\/}}
\def\head{\emph{Head\/}}
\def\body{\emph{Body\/}}
\def\val#1#2{\emph{val\,}_{#1}({#2})}

\newcommand{\I}{\mathcal{I}}
\newcommand{\J}{\mathcal{J}}
\newcommand{\M}{\mathcal{M}}

\newcommand{\PP}{\mathcal{P}}

\let\svprovproof\proof
\let\svprovendproof\endproof
\let\proof\relax
\let\endproof\relax
\let\proof\svprovproof
\let\endproof\svprovendproof
\makeatletter
\newcommand{\printfnsymbol}[1]{%
  \textsuperscript{\@fnsymbol{#1}}%
}
\makeatother
%

\newcommand{\hti}{\langle \intpropone, \intproptwo \rangle}
\newcommand{\intfolone}{\mathcal{I}}
\newcommand{\intfoltwo}{\mathcal{J}}
\newcommand{\intpropone}{S}
\newcommand{\intproptwo}{S'}
\newcommand{\sigprop}{\Sigma}
\newcommand{\sigfol}{\sigma}
\newcommand{\formfolone}{F}
\newcommand{\formfoltwo}{G}
\newcommand{\formprop}{P}
\newcommand{\setprop}{\mathcal{\formprop}}
\newcommand{\atoms}[1]{A(#1)}
\newcommand{\rulehead}{Hd}
\newcommand{\clhead}{H}

\newcommand{\hta}{HTA^{\omega}}
\newcommand{\fohti}{\langle H, I \rangle}
\newcommand{\ih}{I^H}

\newcommand{\bmid}{\bigm|}

\newcommand{\floor}[1]{\lfloor #1 \rfloor}
\newcommand{\ceil}[1]{\lceil #1 \rceil}
\newcommand{\abs}[1]{\lvert #1 \rvert}
\newcommand{\htinterp}{HT\nobreakdash-interpretation}
\newcommand{\htmodel}{HT\nobreakdash-model}
\newcommand{\fweight}{\mathit{weight}}
\newcommand{\fmin}{\mathit{min}}
\newcommand{\fmax}{\mathit{max}}
\newcommand{\fplussum}{\mathit{sum^+}}
\newcommand{\dia}{\mathbin{\Diamond}}
\newcommand{\grdstd}{gr_{\intfolone}^{\boldp}}
\newcommand{\mg}{mini-\textsc{GRINGO}}

\newcommand{\kc}{{\em k-coloring}\xspace}
\newcommand{\yu}[1]{\todo[backgroundcolor=green!40]{#1}}
\newcommand{\yui}[1]{\todo[backgroundcolor=green!40,inline]{#1}}

\newcommand{\Iu}{\mathcal{I}^{\uparrow}}
\newcommand{\Ju}{\mathcal{J}^{\uparrow}}
\renewcommand{\P}{\mathcal{P}}
\newcommand{\Z}{\mathbb{Z}}
\newcommand{\F}{\mathcal{F}}
\newcommand{\X}{\mathcal{X}}
\newcommand{\XZ}{\X_\Z}
\newcommand{\N}{\mathcal{N}}
\newcommand{\C}{\mathcal{C}}
\newcommand{\PC}{\mathbb{C}}
\newcommand{\PN}{\mathcal{N}}
\newcommand{\Hc}{\mathcal{H}}
\newcommand{\T}{\mathcal{T}}
\newcommand{\taum}{\mathcal{Z}}
\newcommand{\lit}{\mathcal{L}^{inf}_A}
\newcommand{\htmodels}{\models_{ht}}

\newcommand{\taumeral}[1]{\overline{#1}}

\newcommand{\dlv}{\textsc{dlv}}
\newcommand{\smodels}{\textsc{smodels}}

\newcommand{\boldh}{\mathbf{h}}
\newcommand{\boldX}{\mathbf{X}}
\newcommand{\boldx}{\mathbf{x}}
\newcommand{\boldy}{\mathbf{y}}
\newcommand{\boldY}{\mathbf{Y}}
\newcommand{\boldZ}{\mathbf{Z}}
\newcommand{\boldz}{\mathbf{z}}
\newcommand{\boldH}{\mathbf{H}}
\newcommand{\boldi}{\mathbf{i}}
\newcommand{\boldp}{\mathbf{p}}
\newcommand{\boldP}{\mathbf{P}}
\newcommand{\boldq}{\mathbf{q}}
\newcommand{\boldt}{\mathbf{t}}
\newcommand{\boldu}{\mathbf{u}}
\newcommand{\boldU}{\mathbf{U}}
\newcommand{\boldv}{\mathbf{v}}
\newcommand{\boldd}{\mathbf{d}}
\newcommand{\boldc}{\mathbf{c}}
\newcommand{\boldw}{\mathbf{w}}
\newcommand{\boldV}{\mathbf{V}}
\newcommand{\boldW}{\mathbf{W}}
\newcommand{\boldl}{\mathbf{l}}
\newcommand{\boldL}{\mathbf{L}}
\newcommand{\boldr}{\mathbf{r}}

\newcommand{\twodots}{\mathinner {\ldotp \ldotp}}
\newcommand{\REL}{{\rm REL}}
\newcommand{\taups}{\tau^p}
\newcommand{\taup}[1]{\taups(#1)}
\newcommand{\taufos}{\tau^\mathit{fo}}
\newcommand{\taufo}[1]{\taufos(#1)}
\newcommand{\sigmafo}{\sigma_{\mathit{fo}}}
\newcommand{\Pfo}{\P_{\mathit{fo}}}
\newcommand{\Ffo}{\F_{\mathit{fo}}}
\newcommand{\Prop}[1]{\mathit{Prop}(#1)}
\newcommand{\Pred}[1]{\mathit{Pred}(#1)}
\newcommand{\ob}{\mathcal{O}}
\newcommand{\oc}{\mathcal{G}}
\newcommand{\pc}{\mathcal{P}}
\newcommand{\fc}{\mathcal{F}}

\def\i#1{{#1}^i}
\def\n#1{{#1}^n}

\def\sig#1#2{{#1}_{#2}}

\def\citeay#1{\citeauthor{#1}~(\citeyear{#1})}

\newcommand{\srt}{\textbf{S}}
\newcommand{\vars}{\textbf{X}}
\newcommand{\predvars}{\textbf{V}}
\newcommand{\cnst}{\textbf{C}}

\newcommand{\sigsrt}{\Sigma^\textbf{S}}
\newcommand{\sigfc}{\Sigma^\textbf{F}}
\newcommand{\sigpc}{\Sigma^\textbf{P}}
\newcommand{\sigvars}{\Sigma^\textbf{X}}
\newcommand{\sigconst}{\Sigma^\textbf{C}}

\newcommand{\tuples}{\mathit{tuple}}
\newcommand{\sets}{\mathit{set}}
\newcommand{\tuple}[1]{\langle#1\rangle}
\newcommand{\sortuple}{\tuple{s_1,\dotsc,s_n}}
\newcommand{\sortuplep}{\tuple{s_1,\dotsc,s_n,s_{n+1}}}
\newcommand{\farity}{s_1 \times \dots \times s_{n} \rar s_{n+1}}
\newcommand{\parity}{s_1 \times \dots \times s_n}

\newcommand{\far}[3]{#1 \times \dots \times #2 \rar #3}
\newcommand{\farone}[2]{#1 \rar  #2 }
\newcommand{\fartwo}[3]{#1 \times  #2 \rar #3}
\newcommand{\partwo}[2]{#1 \times  #2}
\newcommand{\ruleo}{\;\hbox{:-}\;}

\newcommand{\fit}{fit}

\newcommand{\fsum}{\mathit{sum}}
\newcommand{\fcount}{\mathit{count}}
\newcommand{\pmem}{\in}
\newcommand{\pmemF}[2]{#1 \in #2}
\newcommand{\fsub}{\mathit{rem}}
\newcommand{\ftuplen}[1]{\mathit{tuple}_{#1}}
\newcommand{\ftuple}{\ftuplen{k}}
\newcommand{\ffirst}{\mathit{weight}}

\newcommand{\ttuple}{t_{\mathit{tuple}}}
\newcommand{\tset}{t_{\mathit{set}}}

\newcommand{\dtuple}{d_{\mathit{tuple}}}
\newcommand{\dset}{d_{\mathit{set}}}
\newcommand{\eset}{e_{\mathit{set}}}

\newcommand{\sortsetint}{s_{\mathit{setint}}}
\newcommand{\sorttuple}{s_{\mathit{tuple}}}
\newcommand{\sortint}{s_{\mathit{i}}}
\newcommand{\sortset}{s_{\mathit{set}}}
\newcommand{\sortsuper}{s_{\mathit{p}}}
\newcommand{\sortlike}{s_{\mathit{i}}}
\newcommand{\universe}[2]{|#1|^{#2}}

\newcommand{\varsetint}{v_{\mathit{setint}}}
\newcommand{\vartuple}{T}
\newcommand{\varint}{N}
\newcommand{\varset}{S}
\newcommand{\varsuper}{X}

\newcommand{\sm}{\ensuremath{\text{\rm SM}}}
\newcommand{\smp}{\ensuremath{\sm_\boldp}}

\newcommand{\sumaxioms}{\ensuremath{\Delta_{\Sigma}}}
\newcommand{\remaxioms}{\ensuremath{\Delta_{R}}}

\newcommand{\setaxioms}{\ensuremath{\Delta_{S}}}
\newcommand{\aggaxioms}{\ensuremath{\Delta_{A}}}
\newcommand{\uinterpretation}{\mbox{u-interpretation}\xspace}
\newcommand{\uinterpretations}{\mbox{u-interpretations}\xspace}
\newcommand{\htinterpretation}{\mbox{ht-interpretation}\xspace}
\newcommand{\htinterpretations}{\mbox{ht-interpretations}\xspace}
\newcommand{\pstable}{$\boldp$\nobreakdash-stable\xspace}
\newcommand{\infpstable}{INF\nobreakdash-$\boldp$\nobreakdash-stable\xspace}

\newcommand{\IAA}{\tuple{\Hc,\T}}
\newcommand{\IA}{\tuple{\Hc,\T}}
\newcommand{\SM}{\ensuremath{\text{\rm SM}}}
\newcommand{\COMP}{\ensuremath{\text{\rm COMP}}}

\newcommand{\eqdef}{%
  \mathrel{\vbox{\offinterlineskip\ialign{%
    \hfil##\hfil\cr%
    $\scriptscriptstyle\mathrm{def}$\cr%
    \noalign{\kern1pt}%
    $=$\cr%
    \noalign{\kern-0.1pt}%
}}}}

\newcommand{\IofTSP}{I_{|\sigma({\Pi_{\langle E,cst,m \rangle}})}}

\newcommand{\aggstandard}{aggregate}
\newcommand{\Infinite}[1]{\mathit{Infinte}(#1)}
\newcommand{\Inf}[1]{\Infinite{#1}}
\newcommand{\Finite}[1]{\mathit{Finte}(#1)}
\newcommand{\FiniteW}[1]{\mathit{FiniteWeight}(#1)}
\newcommand{\ZeroWeight}[1]{\mathit{ZeroWeight}(#1)}
\newcommand{\Sum}[1]{\mathit{Sum}(#1)}
\newcommand{\Ans}[1]{\mathit{Ans}(#1)}

\newcommand{\domainset}{d_{\mathit{set}}}
\newcommand{\domainsetB}{e_{\mathit{set}}}
\newcommand{\domaintuple}{d_{\mathit{tuple}}}
\newcommand{\domaintupleB}{c_{\mathit{tuple}}}
\newcommand{\agginterp}{agg\nobreakdash-interpretation}
\newcommand{\agginterps}{\agginterp s}
\newcommand{\sigmaagg}{\sigma_{\mathit{agg}}}
\newcommand{\sigmaext}{\sigma_{\mathit{ext}}}
\newcommand{\FiniteF}{f_{|E|}}
\newcommand{\InfiniteF}{g_{|E|}}
\newcommand{\Injective}{\emph{Injective}}
\newcommand{\InjectiveW}{\emph{InjectiveWeight}}
\newcommand{\Image}{\emph{Image}}
\newcommand{\ImageW}{\emph{ImageWeight}}
\newcommand{\SubSet}{\emph{Subset}}

\newcommand{\FiniteSum}{\emph{FiniteSum}}
\newcommand{\Count}{\emph{FiniteCount}}
\newcommand{\Natural}{\mathbb{N}}
\newcommand{\numeral}[1]{\overline{#1}}

\newcommand{\HF}{\mathcal{H}}
\newcommand{\Int}{I}
\newcommand{\grp}[3]{\mathit{gr}^{#1}_{#2}(#3)}
\newcommand{\gr}[2]{\grp{\boldp}{#1}{#2}}
\newcommand{\SPos}{{\rm Pos}}
\newcommand{\Pnn}{{\rm Pnn}}
\newcommand{\Nnn}{{\rm Nnn}}
\newcommand{\Rules}{{\rm Rules}}
\newcommand{\Choice}{{\rm Choice}}
\newcommand{\Top}[2]{\mathit{t}(#1,#2)}
\newcommand{\Bottom}[2]{\mathit{b}(#1,#2)}
\newcommand{\TopA}[2]{\mathit{ta}(#1,#2)}
\newcommand{\BottomA}[2]{\mathit{ba}(#1,#2)}
\newcommand{\TopR}[2]{\mathit{tr}(#1,#2)}
\newcommand{\BottomR}[2]{\mathit{br}(#1,#2)}

\def\t{\paragraph}
\def\G{\Gamma}
\def\rar{\rightarrow}
\def\lrar{\leftrightarrow}
\def\o{\overline}
\def\r#1#2{\frac{\textstyle #1}{\textstyle #2}}
\def\ba{\begin{array}}
\def\ea{\end{array}}
\def\bce{\begin{center}}
\def\ece{\end{center}}
\def\pr#1{\medskip\noindent{\bf #1.}}
\def\beq{\begin{equation}}
\def\eeq#1{\label{#1}\end{equation}}
\def\seq{\Rightarrow}
\def\LRa{\equiv}
\def\false{\bf false}
\def\true{\bf true}
\newcommand{\sgn}{\mathop{\mathrm{sgn}}}

%% file: content.tex
\section{Syntax of Logic Programs}
\label{sec:syntax}
We now present the language of logic programs considered in this paper. It can be viewed as a fragment of the Abstract Gringo (AG) language~\cite{geb15};
equivalently, it can be viewed as an extension of the \mg\ language~\cite{lif19} to rules whose bodies may contain conditional literals.

We assume a {\em (program) signature} with three countably infinite sets of symbols:
\emph{numerals}, \emph{symbolic constants} and \emph{variables}.
We also assume a 1-to-1 correspondence between numerals and integers; the numeral corresponding to an integer~$n$ is denoted by~$\numeral{n}$.
A syntactic expression is~\emph{ground} if it contains no variables.
A ground expression is {\em precomputed} if it contains no operation names.
{\em Terms} are defined recursively:
\begin{itemize}
    \item Numerals, symbolic constants, variables, or either of the special symbols~$\mathit{inf}$ and~$\mathit{sup}$ are terms;
    \item if $t_1$, $t_2$ are program terms and $\circ$ is one of the {\em operation names}
        \begin{gather}
            \label{extended:arithmetic}
            + \quad - \quad \times \quad / \quad \setminus \quad ..
        \end{gather}
        then $t_1 \circ t_2$ is a term (we write $-t$ to abbreviate the term $\overline{0} - t$);
    \item if $t_1$ is a program term, then $\lvert t_1 \rvert$ is a term.
\end{itemize}
We assume that a total order on ground terms is chosen such that
\begin{itemize}
	\item \strut $\mathit{inf}$ is its least element and $\mathit{sup}$ is its greatest element,
	\item for any integers~$m$ and $n$, $\numeral{m} < \numeral{n}$ iff $m < n$, and
	\item for any integer~$n$ and any symbolic constant~$c$, $\numeral{n} < c$.
\end{itemize}
A comparison is an expression of the form~$t_1 \prec t_2$, where $t_1$ and $t_2$ are terms and~$\prec$ is one of the comparison symbols:
\begin{align}
	\label{rel}
	= \quad \neq \quad < \quad > \quad \leq \quad \geq
\end{align}
An \emph{atom} is an expression of the form~$p(\boldt)$, where $p$~is a symbolic constant and~$\boldt$ is a list of program terms.
A \emph{basic literal} is an atom possibly preceded by one or two occurrences of \emph{not}.
A {\em conditional literal} is an expression of the form
$$\clhead\,:\,l_1,\dots,l_m,$$ 
where $\clhead$ is either a comparison, a basic literal, or the symbol $\bot$ (denoting falsity) and $l_1,\dots,l_m$ is a list of basic literals and comparisons. 
We often abbreviate such an expression as $\clhead : \boldL$. 
If $m = 0$, then the preceding ``$:$'' is dropped (so that the program stays \clingo-compliant~\cite{geb15}). We view
 basic literals and comparisons  as  conditional literals with an empty list of conditions, i.e., $m=0$.
A \emph{rule} is an expression of the form
\begin{align}
	\label{rule}
	\rulehead \ruleo B_1, \dots, B_n,
\end{align}
where
\begin{itemize} 
    \item $\rulehead$ is either an atom (a normal rule),
    or an atom in braces (a choice rule),
    or the symbol $\bot$ (a constraint);
	\item each $B_i$ ($1\leq i\leq n$) is a conditional literal.
\end{itemize}
We call the symbol $\hbox{:-}$ a {\em rule operator}. 
We call the left hand side of the rule operator the {\em head}, the right hand side of the rule operator the {\em body}. 
The symbol $\bot$ may be omitted from the head, resulting in an empty head. Such rules are called {\em constraints}.
If the body of the rule is empty, the rule operator will be omitted, resulting in a {\em fact}.
A {\em program} is a finite set of rules.
%

\section{Semantics of Logic Programs via the SM Operator}
\label{sec:semantics:sm}

Here, we introduce the $\SM$ operator-based semantics for logic programs written in the syntax of Section~\ref{sec:syntax}. 
Subsections reviewing necessary concepts are prefixed by the word {\em preliminaries}.
The introduction of these semantics is split into two parts. 
The first part is given in Section~\ref{sec:part1}, where a translation from a logic program to a many-sorted first order theory is provided. This section builds on earlier translations by Fandinno et al.~\cite{falilusc20a,fanliftem24} and Hansen and Lierler~\cite{hanlie22}.
Section~\ref{sec:sm} provides us with the details of the second part, where we start by reviewing the $\SM$ operator and conclude with the definition of answer sets for the considered logic programs. These programs are translated into first-order theories and then the $\SM$ operator is applied. 
Certain models of the resulting formula that we call ``standard'' correspond to answer sets.

\subsection{Translation $\tau^*$ Extended}\label{sec:part1}

In this section, we introduce an extension of the $\tau^*$ translation from logic programs to many-sorted first-order theories (we refer to the introduced extension with the same symbol $\tau^*$).
This extension combines elements of the most recent incarnation of the~$\tau^*$ transformation~\cite{fanliftem24} with the translation $\phi$ for conditional literals~\cite{hanlie22}. 
Following the example of past work~\cite{fan22}, we extend the~$\tau^B$ component of the $\tau^*$ translation with special treatment for global variables. 
%

\subsubsection{Preliminaries: The Target Language of $\tau^*$} \label{sec:tau:target}

A {\em signature} $\sigfol$ consists of \emph{function} and \emph{predicate} constants in addition to a set of \emph{sorts}. 
For every sort $s$, a many-sorted interpretation~$\intfolone$ has a non-empty universe~$\universe{\intfolone}{s}$ (we further assume that there are infinitely many variables for each sort).
A reflexive and transitive \emph{subsort} relation $\prec$ is defined on the set of sorts such that 
when sort $s_1 \prec s_2$, an interpretation $\intfolone$ satisfies the condition~$\universe{\intfolone}{s_1} \subseteq \universe{\intfolone}{s_2}$.
The \emph{function signature} of every function constant $f$ consists of a tuple of~\emph{argument sorts} $s_1,\dots,s_n$, and~\emph{value sort} $s$, denoted by
$
\farity.
$

\emph{Object constants} are function constants with $n=0$, their function signature contains only a value sort.
Similarly, the \emph{predicate signature} of every predicate constant $p$ is a tuple of argument sorts $\parity$.
A predicate constant whose predicate signature is the empty tuple is called a~\emph{proposition}.
The \emph{arity} of a function or predicate signature with $n$ argument sorts is $n$.

\emph{Terms} of a signature $\sigma$ are constructed recursively from function constants.
Atomic formulas are built similar to the standard unsorted logic with the restriction that in
a term $f(t_1,\dots,t_n)$ (an atom $p(t_1,\dots,t_n)$, respectively), the sort of term $t_i$ must be a subsort of the i-th argument of~$f$ (of~$p$, respectively).
In addition,~$t_1 = t_2$ is an atomic formula if the sorts of~$t_1$ and~$t_2$ have a common supersort.
The notion of satisfaction is analogous to the unsorted case with the restriction that an interpretation
maps a term to an element in the universe of its associated~sort.


Our translation $\tau^*$ transforms a logic program $\Pi$ written in the syntax of Section~\ref{sec:syntax} into a first-order sentence with equality over a signature~$\sigfol_\Pi$ of {\em two sorts}.
The first sort is called the \emph{program sort} (denoted $\sortsuper$);  all program terms are of this sort.
The second sort is called the {\em integer sort} (denoted $\sortint$); it is a subsort of the program sort. 
Specifically, the variables of $\sortint$ range over numerals.
Variants of~$X,Y,Z$ will denote variables of sort $\sortsuper$ and variants of~$I,J,M,N$ will denote integer variables.
Bold face variants will denote lists of variables.
To define the remainder of this signature, we must introduce the concepts of \emph{occurrences} and \emph{global variables}.

A {\em predicate symbol} is a pair $p/n$, where $p$ is a symbolic constant and $n$ is a nonnegative integer.
About a program or other syntactic expression, we say that a predicate symbol $p/n$ {\em occurs} in it if it contains an atom of the form $p(t_1,\dots,t_n)$.

A variable is {\em global} in a conditional literal $\clhead : \boldL$ if it occurs in $\clhead$ but not in $\boldL$. 
Thus, any variables in basic literals or comparisons are also global.
A variable is global in a rule if it is global in any of the rule's expressions.

For a program $\Pi$, signature~$\sigma_\Pi$ contains:
\begin{enumerate}
\item\label{en:terms} all precomputed terms as object constants of the program sort; a precomputed constant is assigned the sort $\sortint$ iff it is a numeral;

\item \label{en:preds} all predicate symbols occurring in~$\Pi$ as predicate constants with all arguments of the sort $\sortsuper$;

\item\label{en:1:3} the comparison symbols other than equality and inequality as predicate constants with predicate signature $\partwo{\sortsuper}{\sortsuper}$ (we will use infix notation for constructing these atoms);

\item\label{en:arithmetic} function constants $+$, $-$, and $\times$ with function signature \hbox{$\fartwo{\sortint}{\sortint}{\sortint}$}, and function constant $\lvert \cdot \rvert$ with function signature \hbox{$\farone{\sortint}{\sortint}$};
\end{enumerate}


\subsubsection{Preliminaries: Values of Terms} \label{sec:prelim:vt}
In the language of Section~\ref{sec:syntax}, a term may have one value (as in $3 + 5$), many values (as in $3 + 1..5$), or no values (as in $a + 3$).
Thus, for every program term $t$, we define a formula $val_t(Z)$, where $Z$ is a program variable with no occurrences in $t$.
It indicates that $Z$ is a value of $t$.
\begin{itemize}
    \item if $t$ is a numeral, symbolic constant, program variable, $\mathit{inf}$, or $\mathit{sup}$, then $val_t(Z)$ is $Z = t$;
    \item if $t$ is $\vert t_1 \rvert$, then $val_t(Z)$ is $\exists I (val_{t_1}(I) \wedge Z = \vert I \rvert)$;
    \item if $t$ is $(t_1 \circ t_2)$, where $\circ$ is one of $+$, $-$, or $\times$, then  $val_t(Z)$ is
    \begin{gather*}
        \exists IJ(Z = I \circ J \wedge val_{t_1}(I) \wedge val_{t_2}(J))
    \end{gather*}
    where $I$, $J$ are fresh integer variables;
    \item if $t$ is $t_1/t_2$ then $val_t(Z)$ is
    $$
    \exists I J K (val_{t_1}(I) \wedge val_{t_2}(J) \wedge F_1(IJK) \wedge F_2(IJKZ))
    $$
    where $F_1(IJK)$ is 
    $$K \times \lvert J \rvert \leq \lvert I \rvert < (K + \overline{1}) \times \lvert J \rvert$$
    and $F_2(IJKZ)$ is 
    $$(I \times J \geq \overline{0} \wedge Z = K) \vee (I \times J < \overline{0} \wedge Z = -K)$$
    \item if $t$ is $t_1 \setminus t_2$ then $val_t(Z)$ is
    $$
    \exists I J K (val_{t_1}(I) \wedge val_{t_2}(J) \wedge F_1(IJK) \wedge F_3(IJKZ))
    $$
    where $F_3(IJKZ)$ is
    $$(I \times J \geq \overline{0} \wedge Z = I - K \times J) \vee (I \times J < \overline{0} \wedge Z = I + K \times J)$$
    \item if $t$ is $t_1..t_2$ then $val_t(Z)$ is 
    $$\exists I J K (Z = K \wedge I \leq K \leq J \wedge val_{t_1}(I) \wedge val_{t_2}(J))$$
    where $I \leq K \leq J$ is an abbreviation for $I \leq K \wedge K \leq J$.
\end{itemize}

\subsubsection{Translation $\tau^*$}\label{sec:translationtaustar}
We now describe a translation~$\tau^*$ that converts a program  into a finite set of first-order sentences.
It will be helpful to consider additional notation.
For a tuple of terms $t_1,\dots,t_k$, abbreviated as $\bold{t}$, and a tuple of variables $V_1,\dots,V_k$, abbreviated as $\bold{V}$, we use $val_{\bold{t}}(\bold{V})$ to denote the formula 
$$val_{t_1}(V_1) \wedge \dots \wedge val_{t_k}(V_{t_k}).$$
%
%
Now we introduce $\tau^B_\boldZ$.
It extends the $\tau^B$ translation~\cite{fanliftem24} with a translation for conditional literals.
We use the $\boldZ$ subscript to denote the set of global variables present in a rule.
Given a list~$\boldZ$  of global variables in some rule $R$, we define~$\tau^B_\boldZ$ for all elements of $R$ as follows:
\begin{enumerate}
    \item $\tau^B_\boldZ(\bot)$ is~$\bot$;
    \item $\tau^B_\boldZ(p(\bold{t}))$ is 
        $\exists \bold{V} (val_{\bold{t}}(\bold{V}) \wedge p(\boldV))$
        for every basic literal $p(\boldt)$;
    \item $\tau^B_\boldZ(not ~p(\bold{t}))$ is 
        $\exists \bold{V} (val_{\bold{t}}(\bold{V}) \wedge \neg p(\boldV))$
        for every basic literal $not ~p(\boldt)$;
    \item $\tau^B_\boldZ(not ~not ~p(\bold{t}))$ is 
        $\exists \bold{V} (val_{\bold{t}}(\bold{V}) \wedge \neg \neg p(\boldV))$
        for every basic literal $not ~not ~p(\boldt)$
    \item $\tau^B_\boldZ(t_1 \prec t_2)$ is 
        $\exists Z_1 Z_2(val_{t_1}(Z_1) \wedge val_{t_2}(Z_2) \wedge Z_1 \prec Z_2)$
        for every comparison $t_1 \prec t_2$;
    \item $\tau^B_\boldZ (\boldL)$ is 
        $\tau^B_\boldZ(l_1) \wedge \dots \wedge \tau^B_\boldZ(l_m)$
        for a  list $\boldL$ of basic literals and comparisons; 
    \item $\tau^B_\boldZ(\clhead : \boldL)$ is 
    $$
    \forall \boldX \left(\tau^B_\boldZ(\boldL) \rar \tau^B_\boldZ(\clhead)\right)
    $$
    for every conditional literal $H : \boldL$ occurring in the body of $R$, 
    where $\boldX$ is the set of variables occurring in $H : \boldL$ that do not occur in $\boldZ$.
\end{enumerate}

In what follows, for each rule $R$, $\boldZ$ denotes the list of the global variables of~$R$, and~$\bold{V}$ denotes a list of fresh, alphabetically first program variables.
We now define the translation~$\tau^*$. 
\begin{enumerate}
    \item For a basic rule~$R$ of the form
    $
    p(\bold{t}) \ruleo B_1, \dots, B_n
    $,
    its translation~$\tau^* R$ is
    $$
    \widetilde{\forall} \left(val_{\boldt}(\boldV) \wedge \tau^B_\boldZ(B_1) \wedge \dots \wedge \tau^B_\boldZ(B_n) \to p(\boldV)\right).
    $$
    \item For a choice rule~$R$ of the form
    $
    \{p(\boldt)\} \ruleo B_1, \dots, B_n
    $,
    its translation~$\tau^* R$ is 
    $$
    \widetilde{\forall} \left(val_{\boldt}(\boldV) \wedge \tau^B_\boldZ(B_1) \wedge \dots \wedge \tau^B_\boldZ(B_n) \wedge \neg\neg p(\boldV) \to p(\boldV)\right).
    $$
    \item For a constraint~$R$ of the form
    $
    \bot \ruleo B_1, \dots, B_n
    $,
    its translation~$\tau^* R$ is
    $$
    \forall \boldZ \big(\tau^B_\boldZ(B_1) \wedge \dots \wedge \tau^B_\boldZ(B_n) \to \bot\big).
    $$
    \item For every program~$\Pi$, its translation~$\tau^* \Pi$ is the first-order theory containing~$\tau^* R$ for each rule~$R$ in~$\Pi$.
\end{enumerate}
%

\begin{example}
    For a list of global variables $\{V\}$, $\tau^B_{\{V\}} \left(not ~asg(V, ~C) ~: ~col(C)\right)$ is
    \begin{align*}
        \forall C (\exists Z (Z = C \wedge col(Z)) \rar \exists Z Z_1 (Z = V \wedge Z_1 = C \wedge \neg asg(Z, Z_1)))
    \end{align*}
Thus, the translation of rule~\eqref{eq:rule.color.cl.constraint} is
    \begin{align*}
        \forall V (\tau^B_{\{V\}} \left(not ~asg(V, ~C) ~: ~col(C)\right) \wedge \exists Z (Z = V \wedge vtx(Z)) \rar \bot).
    \end{align*}
\end{example} 

\subsection{Semantics via Many-sorted SM}
\label{sec:sm}

\subsubsection{Preliminaries: The \textit{\SM} Operator for Many-sorted Signatures}
\label{sec:sm:def}

This subsection reviews an extension of the $\SM$ operator~\cite{ferleelif11a} to the many-sorted setting~\cite{fan22}.
This operator is applied to a set of (many-sorted) first-order sentences corresponding to the \textit{formula representation} of a logic program to obtain a set of (many-sorted) second-order sentences.
Models of this set respecting certain assumptions (such as a Herbrand interpretation of symbolic constants) capture the stable models of the original logic program.
We  use $\tau^*$ to obtain the formula representation of a program in a specific many-sorted signature $\sigfol_\Pi$, and apply $\SM$ to the result to characterize stable models of the program.

If $p$ and $u$ are predicate constants or variables with the same predicate signature, then $u \leq p$ stands for
the formula
$$\forall{\boldW}(u(\boldW) \rar p(\boldW)),$$
where~$\boldW$ is an $n$\nobreakdash-tuple of distinct object variables.
If $\boldp$ and $\boldu$ are tuples $p_1,\dots,p_n$ and $u_1,\dots,u_n$ of predicate constants or variables such that each~$p_i$ and~$u_i$ have the same predicate signature, then $\boldu \leq \boldp$ stands for the conjunction 
$(u_1 \leq p_1) \wedge \dots \wedge (u_n \leq p_n),$
and $\boldu < \boldp$ stands for $(\boldu \leq \boldp) \wedge \neg(\boldp \leq \boldu)$. 
For any many\nobreakdash-sorted first\nobreakdash-order formula $\formfolone$ and a list~$\boldp$ of predicate constants, by $\smp[\formfolone]$ we denote the second\nobreakdash-order formula
$$
\formfolone \wedge \neg\exists{\mathbf{u}}\big((\mathbf{u} < \mathbf{p}) \wedge \formfolone^*(\mathbf{u})\big)
$$
where \textbf{u} is a list of distinct predicate variables $u_1,\dots,u_n$ of the same length as~$\boldp$, such that the predicate signature of each~$u_i$ is the same as the predicate signature of $p_i$, and $\formfolone^*(\mathbf{u})$ is defined recursively:
\begin{itemize}
    \item $\formfolone^* = \formfolone$ for any atomic formula $\formfolone$ that does not contain members of \textbf{p},
    \item $p_i(\mathbf{t})^* = u_i(\mathbf{t})$ for any predicate symbol~$p_i$ belonging to~$\boldp$ and any list $\textbf{t}$ of terms,
    \item $(\formfolone \wedge \formfoltwo)^* = \formfolone^* \wedge \formfoltwo^*$,
    \item $(\formfolone \vee \formfoltwo)^* = \formfolone^* \vee \formfoltwo^*$,
    \item $(\formfolone \rar \formfoltwo)^* = (\formfolone^* \rar \formfoltwo^*) \wedge (\formfolone \rar \formfoltwo)$,
    \item $(\forall{x}\formfolone)^* = \forall{x}\formfolone^*$,
    \item $(\exists{x}\formfolone)^* = \exists{x}\formfolone^*$.
\end{itemize}

\begin{definition}
    \label{def:p.stable.models}
    For a many-sorted first-order sentence $F$ from signature $\sigma$, the models of $SM_{\bold{p}}[F]$ are called the $\bold{p}$\nobreak-stable models of $F$.
    For a set $\Gamma$ of first-order sentences, the $\bold{p}$\nobreak-stable models of $\Gamma$ are the $\bold{p}$\nobreak-stable models of the conjunction of all formulas in $\Gamma$.
\end{definition}
The list $\boldp$ of predicates of a $\bold{p}$\nobreak-stable model are called {\em intensional} -- ``belief'' in these predicates is minimized.
Predicates that are not intensional are called {\em extensional}.

\subsubsection{Answer Sets via Standard Interpretations}
Recall from Section~\ref{sec:tau:target} that $\tau^*$ maps a program $\Pi$ with intensional predicates $\boldp$ into a formula within signature $\sigfol_\Pi$ of two sorts.
We now consider a special type of interpretation of $\sigfol_\Pi$ (a \emph{standard} interpretation) to restrict the models of $\SM_\boldp[\tau^*\Pi]$ to exactly the stable models of $\Pi$.

A standard interpretation $\I$ satisfies that
\begin{enumerate}
    \item the universe $|\I|^{\sortsuper}$ is the set of all precomputed terms;
    \label{c.interp.first}
    \item the universe $|\I|^{\sortint}$ is the set of all numerals;
    \label{c.interp.second}
    \item $\I$ interprets every precomputed term $t$ as $t$;
    \label{c.interp.third}
    \item $\I$ interprets $\numeral{m} + \numeral{n}$ as~$\numeral{m+n}$, and similarly for subtraction and multiplication;
    \label{c.interp.fourth}
    \item $\I$ interprets $\lvert \numeral{n} \rvert$ as~$\numeral{\lvert n \rvert}$;
    \label{c.interp.fifth}
    \item $\I$ interprets every comparison $t_1 \prec t_2$, where $t_1$ and $t_2$ are precomputed terms, as true iff the relation $\prec$ holds for the pair $(t_1,t_2)$.
    \label{c.interp.seventh}
\end{enumerate}
For a standard interpretation $\intfolone$, we use $\atoms{\intfolone}$ to denote the unique set of precomputed atoms to which $\intfolone$ assigns the value \emph{true}.

\begin{definition}
    \label{def:answer.sets}
    Let $\Pi$ be a  program and let $\bold{p}$ be a list of some predicate symbols occurring in $\Pi$ other than the comparison symbols.
    Then, for a standard interpretation $\intfolone$, we call $\atoms{\intfolone}$ a {\em $\bold{p}$-answer set} of $\Pi$ when $\intfolone$ is a $\bold{p}$\nobreak-stable model of $\tau^*\Pi$.
    When $\bold{p}$ is the list of all predicate symbols occurring in $\Pi$ other than the comparison symbols, then we simply call a {\em $\bold{p}$-answer set} of $\Pi$ an {\em answer set} of $\Pi$.
\end{definition}
In the sequel, we illustrate how answer sets as defined here coincide with the notion of what we later call gringo answer sets. Yet, it is interesting to note that $\SM$\nobreakdash-operator based semantics enable a more flexible understanding of a $\bold{p}$-answer set that distinguishes intensional (namely, $\bold{p}$)  and extensional predicate symbols.

\section{Review: Semantics of Logic Programs via Infinitary Propositional Logic}
\label{sec:semantics:ipl}
Truszczy\'{n}ski (2012) provides a characterization of logic program semantics in terms of infinitary propositional logic (IPL) and illustrates the precise relation  to the $\SM$-based characterization~\cite{truszczynski12a}.
IPL is a useful technical device due to its close relation to the grounding  procedure implemented by answer set solver \clingo.
In this section, we review IPL and the infinitary logic of here-and-there, which are important tools for establishing the results of this paper. 
Unless otherwise specified, $\formprop$ and its variants are used to denote IPL formulas, whereas $\setprop$ and its variants denote sets of IPL formulas. Finally, we introduce \emph{gringo answer sets}, which constitute the established semantic characterization of logic programs with conditional literals and arithmetic.

\subsection{Infinitary Formulas}
A {\em propositional signature} $\sigprop$ is a set of propositional atoms.
For every nonnegative integer $r$, {\em (infinitary) formulas} over $\sigprop$ of rank $r$ are defined recursively:
\begin{itemize}
    \item every atom from $\sigprop$ is a formula of rank 0,
    \item if $\setprop$ is a set of formulas, and $r$ is the smallest nonnegative integer that is greater than the ranks of all elements of $\setprop$, then $\setprop^\wedge$ and $\setprop^\vee$ are formulas of rank $r$,
    \item if $\formprop$ and $\formprop'$ are formulas, and $r$ is the smallest nonnegative integer that is greater than the ranks of $\formprop$ and $\formprop'$, then $\formprop \to \formprop'$ is a formula of rank $r$.
\end{itemize}
$\formprop \wedge \formprop'$ is shorthand for $\{\formprop,\formprop'\}^\wedge$, and $\formprop \vee \formprop'$ is shorthand for $\{\formprop,\formprop'\}^\vee$.
We use $\top$ and $\bot$ as abbreviations for $\emptyset^\wedge$ and $\emptyset^\vee$, respectively.
Further, $\neg \formprop$ stands for $\formprop \to \bot$, and $\formprop \leftrightarrow \formprop'$ stands for $(\formprop \to \formprop') \wedge (\formprop' \to \formprop)$.

A \emph{propositional interpretation} of $\sigprop$ is a subset $\intpropone$ of $\sigprop$.
The {\em satisfaction relation} between a propositional interpretation and an infinitary  formula is defined recursively:
\begin{itemize}
    \item For every atom $p$ from $\sigprop$, $\intpropone \models p$ if $p \in \intpropone$.
    \item $\intpropone \models \setprop^\wedge$ if, for every formula $\formprop$ in $\setprop$, $\intpropone \models \formprop$.
    \item $\intpropone \models \setprop^\vee$ if there is a formula $\formprop$ in $\setprop$ such that $\intpropone \models \formprop$.
    \item $\intpropone \models \formprop \to \formprop'$ if $\intpropone \not\models \formprop$ or $\intpropone \models \formprop'$.
\end{itemize}

\subsection{Infinitary Logic of Here-and-there and Truszczy\'{n}ski Stable Models}
\label{sec:ipl:tsm}
An \emph{\htinterp} of propositional signature $\sigprop$ is an ordered pair $\hti$ of interpretations of $\sigprop$ such that $\intpropone \subseteq \intproptwo$.
The {\em satisfaction relation (ht-satisfaction)} between an \htinterp\ and an infinitary formula is defined recursively:
\begin{itemize}
    \item For every atom $p$ from $\Sigma$, $\hti \models p$ if $p \in S$.
    \item $\hti \models \setprop^\wedge$ if, for every formula $\formprop$ in $\setprop$, $\hti \models \formprop$.
    \item $\hti \models \setprop^\vee$ if there is a formula $\formprop$ in $\setprop$ such that $\hti \models \formprop$.
    \item $\hti \models \formprop \to \formprop'$ if
    \begin{enumerate}
        \item $S' \models \formprop \to \formprop'$, and
        \item $\hti \not\models \formprop$ or $\hti \models \formprop'$.
    \end{enumerate}
\end{itemize}
An \emph{\htmodel} of an infinitary formula  is an \htinterp\ that satisfies this formula.
An \emph{\htmodel} of a set $\setprop$ of infinitary formulas is an \htinterp\ that satisfies all formulas in~$\setprop$.
Two infinitary formulas (sets of infinitary formulas) are {\em equivalent} when they have the same {\htmodel}s.
An \htinterp\ of the form $\langle S, S \rangle$ is called \emph{total}.
An \emph{equilibrium model} of a set $\setprop$ of infinitary formulas is a total \htmodel\ $\langle S', S' \rangle$ of~$\setprop$ such that for every proper subset $S$ of $S'$, $\hti$ is not an \htmodel\ of $\setprop$.

%
An interpretation satisfies a set $\setprop$ of formulas (is a model of $\setprop$) if it satisfies all formulas in $\setprop$.

\begin{definition}
    \label{def:tru.stable.models}
    For a set $\setprop$ of infinitary propositional formulas, we say that a propositional interpretation $S$ is a \emph{Truszczy\'{n}ski stable model} if $\langle S, S \rangle$ is an equilibrium model of $\setprop$.
\end{definition}
In the sequel, we may refer to  a {Truszczy\'{n}ski stable model} of an infinitary propositional formula, where we identify this formula with a singleton set containing it.
%



\subsection{Translation $\tau$ and Gringo Answer Sets}
In this section we review the relevant components of $\tau$~\cite{fanliftem24,geb15}.
It is due to note that the 2015 publication is the ``official'' source of the Abstract Gringo semantics and contains a $\tau$ translation for conditional literals, whereas the 2024 publication is more up-to-date with respect to the  definitions of division, modulo, and absolute value adhered to by the latest versions of \clingo.

An expression is {\em ground} if it does not contain variables.
For every ground term $t$, the set of precomputed terms $[t]$ of its {\em values} is defined as follows:
\begin{itemize}
    \item if $t$ is a numeral, symbolic constant, or $\mathit{inf}$ or $\mathit{sup}$ then $[t]$ is $\{t\}$;
    \item if $t$ is $\lvert t_1 \rvert$, then $[t]$ is the set of numerals $\overline{\lvert n \rvert}$ for all integers $n$ such that $\overline{n} \in [t_1]$;
    \item if $t$ is $(t_1 \circ t_2)$, where $\circ$ is one of $+,-,\times$, then $[t]$ is the set of numerals $\overline{n_1 \circ n_2}$ such that $\overline{n_1} \in [t_1]$, $\overline{n_2} \in [t_2]$;
    \item if $t$ is $(t_1/t_2)$, then $[t]$ is the set of numerals $\overline{round(n_1/n_2)}$ for all integers $n_1,n_2$ such that $\overline{n_1} \in [t_1]$, $\overline{n_2} \in [t_2]$ and $n_2 \neq 0$;
    \item if $t$ is $(t_1 \setminus t_2)$, then $[t]$ is the set of numerals $\overline{n_1 - n_2 \cdot round(n_1/n_2)}$ for all integers $n_1,n_2$ such that $\overline{n_1} \in [t_1]$, $\overline{n_2} \in [t_2]$ and $n_2 \neq 0$;
    \item if $t$ is $(t_1..t_2)$, then $[t]$ is the set of numerals $\overline{m}$ for all integers $m$ such that, for some integers $n_1$, $n_2$
    \begin{gather*}
        \overline{n_1} \in [t_1], \quad \overline{n_2} \in [t_2], \quad n_1 \leq m \leq n_2.
    \end{gather*}
    \item if $\boldt$ is a tuple of terms $t_1,\dots,t_n$ ($n > 0$) then $[\boldt]$ is the set of tuples $\langle r_1,\dots,r_n \rangle$ for all $r_1 \in [t_1]$, $\dots$, $r_n \in [t_n]$.
\end{itemize}
The function $round$ is defined as:
\begin{gather}
    round(n) = \left\{
    \begin{array}{ll}
          \floor{n}& \text{if } n \geq 0 \\
          \ceil{ n } & \text{if } n < 0 \\
    \end{array} 
    \right. 
    \label{def:round}
\end{gather}   

A rule (or any other expression in a rule) is called \emph{closed} if it contains no global variables.
An \emph{instance} of a rule~$R$ is any rule that can be obtained from~$R$ by substituting precomputed terms for all global variables.

To transform a closed rule $R$ into a set of infinitary propositional formulas, translation $\tau$ is defined as follows:
\begin{enumerate}[noitemsep,topsep=0pt]
    \item $\tau(\bot)$ is $\bot$;
    \item $\tau(p(\boldt))$ is the disjunction of atoms $p(\boldr)$ over all tuples $\boldr$ in $[\boldt]$ for any ground atom $p(\boldt)$, thus, $\tau(p(\boldt))$ is $p(\boldt)$ if $\boldt$ is a tuple of precomputed terms;
    \item $\tau(not ~p(\boldt))$ is the disjunction of formulas $\neg p(\boldr)$ over all tuples $\boldr$ in $[\boldt]$ for any ground atom $p(\boldt)$;
    \item $\tau(not ~not ~p(\boldt))$ is the disjunction of formulas $\neg \neg p(\boldr)$ over all tuples $\boldr$ in $[\boldt]$ for any ground atom $p(\boldt)$;

    \item $\tau(t_1 \prec t_2)$ is $\top$ if the relation $\prec$ holds between the terms $r_1$, $r_2$ for some $r_1$, $r_2$ such that $r_1 \in [t_1]$ and $r_2 \in [t_2]$, and $\bot$ otherwise;
    \item $\tau (\boldL)$ is $\tau(l_1) \wedge \dots \wedge \tau(l_m)$ for a list $\boldL$ of basic or conditional literals;
    
    \item\label{t:clb} for a closed conditional literal $H : \boldL$ occurring in the body of rule $R$, $\tau(H : \boldL)$ is the conjunction of the formulas 
    $\tau(\boldL^{\boldx}_{\boldr}) \rar \tau(H^{\boldx}_{\boldr})$
    where $\boldx$ is the list of variables occurring in the conditional literal, over all tuples $\bold{r}$ of precomputed terms of the same length as~$\boldx$;
\item for  an instance $\rho$  of 
\begin{itemize}
\item a basic rule $p(\boldt) \ruleo B_1,\dots,B_n$, its translation $\tau \rho$ is 
\begin{gather}
    \label{eq:tau.basic.rule}
    \tau(B_1 \wedge \dots \wedge B_n) \rar \bigwedge_{\boldr \in [\boldt]} p(\boldr);
\end{gather}

\item a choice rule $\{p(\boldt)\} \ruleo B_1,\dots,B_n$, its translation $\tau \rho$ is 
\begin{gather}
    \label{eq:tau.choice.rule}
    \tau(B_1 \wedge \dots \wedge B_n) \rar \bigwedge_{\boldr \in [\boldt]} (p(\boldr) \vee \neg p(\boldr));
\end{gather}
\item  a constraint $ \bot \ruleo B_1,\dots,B_n$,  its translation $\tau \rho$ is 
$$\neg \tau(B_1 \wedge \dots \wedge B_n).$$
\end{itemize}
\color{black}

\end{enumerate}
For any rule $R$ of form~\eqref{rule}, $\tau R$ stands for the conjunction of the elements in the set of formulas that consists of $\tau \rho$ for all instances $\rho$ of $R$.
By definition, each of these instances (and thus the expressions occurring within them) are closed.
For any program~$\Pi$, $\tau\Pi$ is the set of the formulas $\tau R$ for each $R$ in $\Pi$.

\begin{example}
    For instance, $\tau \left( not ~asg(v,C) ~: ~col(C) \right)$ is
    \begin{gather*}
        \{ \left(col(c) \to \neg asg(v,c) \right) \mid c \in |I|^{\sortsuper} \}^\wedge.
    \end{gather*}
    Thus, $\tau$ applied to rule~\eqref{eq:rule.color.cl.constraint} is
    \begin{gather*}
        \{ \neg \left( \tau \left( not ~asg(v,C) ~: ~col(C) \right) \wedge vtx(v) \right) \mid v \in |I|^{\sortsuper}\}^\wedge.
    \end{gather*}
\end{example}

Consider the syntax of the language presented in Section~\ref{sec:syntax}, which is a subset of the Abstract Gringo  language, whose
semantics are defined by the translation~$\tau$.
The following definition is from Gebser et al.~\cite{geb15}.

\begin{definition}
    \label{def:gringo.answer.sets}
    We say that a set  $S$ of ground atoms is a {\em gringo answer set} of a program~$\Pi$  if $S$ is a Truszczy\'{n}ski stable model of $\tau\Pi$.
\end{definition}

\section{Connecting Semantics}
\label{sec:consem} 
Ultimately, this section shows that 
the answer sets as introduced in 
{Definition~\ref{def:answer.sets}} by means of the $\SM$ operator coincide with the gringo answer sets as provided by Definition~\ref{def:gringo.answer.sets}. 
Before that, we review the process of converting first order sentences (typically denoted by variants of $F$ and $G$) into infinitary formulas originally proposed by 
Truszczy\'{n}ski~\cite{truszczynski12a}. In addition, we review the concept of strong equivalence in the settings of IPL. Equipped with these notions, we show how, given a program $\Pi$, the application of the transformation by Truszczy\'{n}ski on $\tau^*\Pi$ results in  an IPL formula that is strongly equivalent to the IPL formula obtained by $\tau\Pi$. This is the key step that helps us to establish our main result: answer sets (Definition~\ref{def:answer.sets}) and  gringo answer sets (Definition~\ref{def:gringo.answer.sets}) coincide.

\subsection{Preliminaries: From First Order Sentences to Infinitary Formulas}

Truszczy\'{n}ski (2012)
provides a definition of stable models for first-order sentences via a grounding procedure, which converts them into infinitary formulas~\cite[Section 3]{truszczynski12a}. 
Later, this grounding process was generalized to many-sorted first-order formulas and extensional predicate symbols~\cite{fan22,falilusc20a}. We review this generalization below. Prior to the review we introduce some necessary notation.

If ${\intfolone}$ is an interpretation of a signature~$\sigma$ then by $\sigma^{\intfolone}$ we denote the signature obtained
from~$\sigma$ by adding, for every element~$d$ of a domain~$|I|^s$, its {\em name}~$d^*$ as an object constant
of sort~$s$.
The interpretation~${\intfolone}$ is extended to $\sigma^{\intfolone}$ by defining~$(d^*)^{\intfolone} = d$.
The value~$t^{\intfolone}$ assigned by an interpretation~${\intfolone}$ of~$\sigma^{\intfolone}$~to a ground term~$t$ over $\sigma^{\intfolone}$
and the satisfaction relation between an interpretation of~$\sigma^{\intfolone}$ and a sentence over~$\sigma^{\intfolone}$  are defined recursively, in the usual way.
If $\boldd$ is a tuple $d_1, \dots, d_n$ of elements of domains of~${\intfolone}$ then $\boldd^*$ stands for the tuple~$d_1^*, \dotsc, d_n^*$ of their names. 
If~$\boldt$ is a tuple~$t_1, \dotsc, t_n$ of ground terms then $\boldt^{\intfolone}$ stands for the tuple~$t_1^{\intfolone}, \dotsc, t_n^{\intfolone}$ of values assigned to them by~${\intfolone}$.


For an interpretation $\intfolone$ and a list $\boldp$ of predicate symbols, by $\intfolone^\boldp$ we denote the set of precomputed atoms $p(t_1,\dots,t_k)$ satisfied by $\intfolone$ where $p \in \boldp$.
Let~$\boldp,\boldq$ be a partition of the predicate symbols in the signature.
Then,
the \emph{grounding of a first-order sentence~$F$ with respect to an interpretation~$\intfolone$ and a set of intensional predicate symbols~$\boldp$} (and extensional predicate symbols~$\boldq$) is defined as follows:
\begin{enumerate}
    \item\label{gr:condition:1} $\gr{\intfolone}{\bot} = \bot$;
    \item\label{gr:condition:2} for~$p \in \boldp$, $\gr{\intfolone}{p(t_1,\dotsc,t_k)} = p((t_1^\intfolone)^*,\dotsc,(t_k^\intfolone)^*)$;
    \item\label{gr:condition:3} for~$p \in \boldq$, $\gr{\intfolone}{p(t_1,\dotsc,t_k)} = \top$ 
    if ${ p((t_1^\intfolone)^*,\dotsc,(t_k^\intfolone)^*) \in I^\boldq}$

    and ${\gr{\intfolone}{p(t_1,\dotsc,t_k)} = \bot}$ otherwise;

    \item\label{gr:condition:4} $\gr{\intfolone}{t_1 = t_2} = \top$ if $t_1^\intfolone = t_2^\intfolone$ and $\bot$ otherwise;

    \item\label{gr:condition:5} $\gr{\intfolone}{F \otimes G} = \gr{\intfolone}{F} \otimes \gr{\intfolone}{G}$ if $\otimes$ is $\wedge$, $\vee$, or $\to$;
    \item\label{gr:condition:6} $\gr{\intfolone}{\exists X \, F(X)} = \{ \gr{\intfolone}{F(u^*)} \mid u \in \universe{\intfolone}{s} \}^{\vee}$ if $X$ is a variable of sort~$s$;
    \item\label{gr:condition:7} $\gr{\intfolone}{\forall X \, F(X)} = \{ \gr{\intfolone}{F(u^*)} \mid u \in \universe{\intfolone}{s} \}^{\wedge}$ if $X$ is a variable of sort~$s$.
\end{enumerate}
Recall that $\neg F$ is an abbreviation for $F \to \bot$ so that
\hbox{$\grdstd(\neg F) =  \neg \grdstd(F)$}.
For a first order theory~$\Gamma$, we define~$\gr{\intfolone}{\Gamma} = \{ \gr{\intfolone}{F} \mid F \in \Gamma \}^\wedge$.
%

In the sequel, most of the references to the stated definition of grounding are in the context of \underline{standard} interpretations. 
Given the conditions~\ref{c.interp.third}-\ref{c.interp.fifth} of the definition of standard interpretations, we can simplify the definition of grounding by restating conditions~\ref{gr:condition:2},~\ref{gr:condition:3},~\ref{gr:condition:6}, and~\ref{gr:condition:7}, as follows   
\begin{enumerate}
    \item[\ref{gr:condition:2}.] for~$p \in \boldp$, $\gr{\intfolone}{p(t_1,\dotsc,t_k)} = p(t_1^\intfolone,\dotsc,t_k^\intfolone)$;
    
    \item[\ref{gr:condition:3}.] for~$p \in \boldq$, $\gr{\intfolone}{p(t_1,\dotsc,t_k)} = \top$ 
    if ${ p(t_1^\intfolone,\dotsc,t_k^\intfolone) \in I^\boldq}$
    
    and ${\gr{\intfolone}{p(t_1,\dotsc,t_k)} = \bot}$ otherwise;

    \item[\ref{gr:condition:6}.] $\gr{\intfolone}{\exists X \, F(X)} = \{ \gr{\intfolone}{F(u)} \mid u \in \universe{\intfolone}{s} \}^{\vee}$ if $X$ is a variable of sort~$s$;
    
    \item[\ref{gr:condition:7}.] $\gr{\intfolone}{\forall X \, F(X)} = \{ \gr{\intfolone}{F(u)} \mid u \in \universe{\intfolone}{s} \}^{\wedge}$ if $X$ is a variable of sort~$s$.
\end{enumerate}




\subsection{Preliminaries: Strong Equivalence in the Infinitary Setting}

The Truszczy\'{n}ski stable models of a set of infinitary logic formulas (Definition~\ref{def:tru.stable.models} in Section~\ref{sec:semantics:ipl}) can be characterized by an extension of equilibrium logic to the infinitary setting~\cite[Theorem 2]{harlifpeval17}.
This allows strong equivalence of infinitary (propositional) formulas to be defined as follows:
About sets $\setprop_1$, $\setprop_2$ of infinitary formulas we say that they are {\em strongly equivalent} (denoted as $\setprop_1 \equiv_s \setprop_2$) to each other if, for every set~$\setprop$ of infinitary formulas, the sets $\setprop_1 \cup \setprop$ and $\setprop_2 \cup \setprop$ have the same 
Truszczy\'{n}ski stable models~\cite{harlifpeval17}.
About infinitary formulas $\formprop$ and $\formprop'$ we say that  they are strongly equivalent if the singleton sets $\{\formprop\}$ and $\{\formprop'\}$ are strongly equivalent.
Theorem 3 from that work shows that two sets  of infinitary formulas are strongly equivalent  if and only if they are equivalent in the infinitary logic of here-and-there.
Sometimes, we will abuse the term {\em strong equivalence} or notation $\equiv_s$ by stating that a set of infinitary formulas is strongly equivalent to an infinitary formula, understanding that in such a case these two entities share the same {\htmodel}s.



\subsection{Connecting Semantics}

In this section, we present the main results of this paper, which relate our proposed semantics  -- the answer sets of Definition~\ref{def:answer.sets} -- to the established semantics -- the gringo answer sets of Definition~\ref{def:gringo.answer.sets}.
Propositions~\ref{prop:gr:val}-\ref{prop:rule} are counterparts of Propositions 1-3 from earlier work by Lifschitz et al.~\cite{lif19} within the context of a different language of logic programs.
In particular,  conditional literals are part of the language considered here. 
We also allow the absolute value function symbol, and our definition of the values of terms of the form $t_1/t_2$ and $t_1 \setminus t_2$ differ. It is due to note that the most recent dialect of \mg~\cite{fanliftem24} uses a different definition of integer division than the one presented in earlier publications.
%

Let us start by reviewing some notation. 
For a program $\Pi$, we call a partition $\boldp,\boldq$ of predicate symbols from the signature $\sigma_\Pi$ (defined in~Section~\ref{sec:tau:target}) the \emph{standard partition} if $\boldp$ contains all predicate symbols from $\Pi$ and $\boldq$ contains every comparison symbol. It is easy to see that the standard partition is unique to program $\Pi$, and thus can be identified with its first element $\boldp$ (all predicate symbols but comparisons occurring in~$\Pi$).
For a formula $F$, variable $Z$, and term $t$, by $F_t^Z$ we denote the result of substituting term $t$ for variable~$Z$. For instance, 
$val_a(Z)_b^Z$ results in formula $b=a$, where $b$ and $a$ are symbolic constants.
In the case when substitution is applied to the formula of the form $val_t(Z)$, we often write $val_t(r)$ to denote $val_t(Z)_r^Z$.
Intuitively, the formula $val_t(r)$ expresses that~$r$ is one of the values of $t$.
This claim is formalized in Proposition~\ref{prop:gr:val} below. 
Note that in the statements of the propositions below, some program $\Pi$ is assumed implicitly with its corresponding signature $\sigma_\Pi$.

\begin{proposition}
    Let $\intfolone$ be a standard interpretation, and let $\boldp$ be the standard partition.
    Then, for any program variable $Z$, ground term $t$, and precomputed term $r$, the formula $\grdstd(val_t(Z)_r^Z)$ is strongly equivalent to $\top$ if $r \in [t]$ and to $\bot$ otherwise.
    \label{prop:gr:val}
\end{proposition}
Recall from Section~\ref{sec:syntax} that in our proposed language, a basic literal or comparison is a special type of conditional literal with an empty list of conditions.
Thus, within our extension of the \mg\ language, conditional literals form the main syntactic element.
Proposition~\ref{prop:gr:val} helps us establish the following result:
\begin{proposition}
    \label{prop:cl}
    Let $\clhead : \boldL$ be a conditional literal, 
    let $\boldZ$ be a superset of the global variables of $\clhead : \boldL$, 
    and let $\boldz$ be a tuple of precomputed terms of the same length as $\boldZ$.
    Then, for a standard interpretation $\intfolone$ and the standard partition $\boldp$, 
    $$
    \grdstd\left( (\tau^B_{\boldZ}\left(\clhead : \boldL\right))^\boldZ_\boldz \right) \equiv_s \tau \left( (\clhead : \boldL)^\boldZ_\boldz \right).
    $$
\end{proposition}
Using the preceding proposition, we can conclude the following result:
\begin{proposition}
    \label{prop:rule}
    For a rule $R$, standard interpretation $\intfolone$, and the standard partition $\boldp$, 
    $$\grdstd(\tau^*(R)) \equiv_s \tau R.$$
\end{proposition}
We can extend Proposition~\ref{prop:rule} to programs:
\begin{proposition}
    \label{prop:program}
    For a program $\Pi$, standard interpretation $\intfolone$, and the standard partition $\boldp$,
    $$\grdstd(\tau^*(\Pi)) \equiv_s \tau \Pi.$$
\end{proposition}
The last proposition is an  important result supporting one of the key theorems of this work, namely, Theorem~\ref{thrm:connecting.tsm}. 
For a (many-sorted) first-order interpretation $\intfolone$ and a set of predicates $\boldp$, we use $\intfolone^\boldp$ to denote $A(\intfolone)$ restricted to atoms whose predicate symbols occur in $\boldp$.
\begin{theorem}
    \label{thrm:connecting.tsm}
    Let $\Pi$ be a program, and let $\boldp$ be the standard partition.
    Then, a set $\I^{\boldp}$ of precomputed atoms is a Truszczy\'{n}ski stable model of $gr_{\I}^{\boldp}(\tau^*\Pi)$
    iff $\I^{\boldp}$ is a Truszczy\'{n}ski stable model of $\tau \Pi$.
\end{theorem}

\begin{proof}
    An immediate consequence of Proposition~\ref{prop:program} is that,
    for a program $\Pi$ and interpretation $\I$, $gr_{\I}^{\boldp}(\tau^*(\Pi)) \equiv_s \tau \Pi$.
    Thus,  $gr_{\I}^{\boldp}(\tau^*(\Pi))$  and $\tau \Pi$ have the same {\htmodel}s.
    This implies that they have the same infinitary equilibrium models.
    Consequently, they have the same Truszczy\'{n}ski stable models due to Theorem~3 by Truszczy\'{n}ski~\cite{truszczynski12a}.
\end{proof}
Now that we have uncovered the connection between the infinitary formulas produced by $gr_{\I}^{\boldp}(\tau^*\Pi)$ and those produced by $\tau\Pi$, we use the result by Fandinno et al. (2020) reformulated below to connect the Truszczy\'{n}ski stable models of these formulas to the $\boldp$-stable models of $\tau^*\Pi$. 

%
\begin{proposition}{\cite[Proposition 2]{falilusc20a}}
    \label{prop:fan20.2}
    For any finite two-sorted (target language) theory $\Gamma$, an interpretation $\intfolone$, and list of predicate symbols $\boldp$, 
    $\intfolone$ is $\boldp$-stable model of $\Gamma$ 
    if and only if
    $\intfolone^\boldp$ is a Truszczy\'{n}ski stable model of $\gr{\intfolone}{\Gamma}$.    
\end{proposition}
We are ready to state the main result of this section, connecting our $\SM$-based semantics for programs with the established semantics reviewed in Section~\ref{sec:semantics:ipl}:
\begin{theorem}
    \label{thrm:main}
    Let $\Pi$ be a program and let $\boldp$ be the standard partition.
    Then, $A(\I)$ is an answer set of $\Pi$ iff $A(\I)$ is a gringo answer set of $\Pi$.
\end{theorem}

\begin{proof}
    
    $A(\I)$ is an answer set of $\Pi$\\
    iff $\I$ is a \pstable model of $\tau^*\Pi$ \hfill (Definition~\ref{def:answer.sets}; 
    $\boldp$ is the standard partition)\\
    iff $\I^\boldp$ is a Truszczy\'{n}ski stable model of $gr_{\I}^{\boldp}(\tau^*\Pi)$ \hfill (Proposition~\ref{prop:fan20.2})\\
    iff $\I^\boldp$ is a Truszczy\'{n}ski stable model of $\tau \Pi$ \hfill (Theorem~\ref{thrm:connecting.tsm})\\
    iff $\I^\boldp$ is a gringo answer set of $\Pi$ \hfill (Definition~\ref{def:gringo.answer.sets})\\
    iff $A(\I)$ is a gringo answer set of $\Pi$ \hfill ($\boldp$ is the standard partition).
\end{proof}

%% file: appendix-1.tex
\subsection{Preliminaries: Basic Facts}

\begin{fact}
    As customary, let $\formprop$ denote an infinitary formula, and $\setprop$ denote a set of infinitary formulas.
    \label{fact:1}
    \begin{enumerate}
        \item\label{l:1:1} $\formprop\wedge \top \equiv_s \formprop$; $\big\{\setprop, \top\big\}^\wedge \equiv_s \setprop^\wedge$.
        \item\label{l:1:2} $\formprop\wedge \bot \equiv_s \bot$; $\big\{\setprop, \bot\big\}^\wedge \equiv_s \bot$.
        \item\label{l:1:3} $\formprop \vee \top \equiv_s \top$; $\big\{\setprop, \top\big\}^\vee \equiv_s \top$.
        \item\label{l:1:4} $\formprop \vee \bot \equiv_s \formprop$; $\big\{\setprop, \bot\big\}^\vee \equiv_s \setprop^\vee$.
        \item\label{l:1:7} $\bot \rar \formprop \equiv_s \top$.
        \item\label{l:1:5} $\setprop \equiv_s \top$ iff every \htinterp\ $\hti$ satisfies $\setprop$.
        \item\label{l:1:6} $\setprop \equiv_s \bot$ iff every \htinterp\ $\hti$ fails to satisfy $\setprop$.
    \end{enumerate}
\end{fact}

\begin{proof}
    The first four points follow directly.
    \\[5pt]
    \\\noindent
    To see point~\ref{l:1:7}, take an arbitrary \htinterp\ $\hti$ and observe that it satisfies $\bot \rar \formprop$.
    By definition of ht-satisfaction, $\hti \models \bot \rar \formprop$ iff $I \models \bot \rar \formprop$ and $\hti \not\models \bot$ or $\hti \models \formprop$.
    $I \models \bot \rar \formprop$ follows from the rules of classical satisfaction, and $\hti \not\models \bot$ follows from the rules of ht-satisfaction.
    \\[5pt]
    \\\noindent
    To see point~\ref{l:1:5}, assume $\setprop \equiv_s \top$ and take an arbitrary \htinterp\ $\hti$.
    $\hti \models \top$ since $\top$ is an abbreviation for $\emptyset^\wedge$.
    $\setprop$ and $\top$ have the same \htmodel s by the definition of strong equivalence, thus $\hti \models \setprop$. Now assume every \htinterp\ $\hti$ satisfies $\setprop$.
    We know that every $\hti \models \top$.
    Thus, $\setprop$ and $\top$ have the same \htmodel s and $\setprop \equiv_s \top$.
    \\[5pt]
    \\\noindent
    To see point~\ref{l:1:6}, assume $\setprop \equiv_s \bot$ and take an arbitrary \htinterp\ $\hti$.
    $\hti \not\models \bot$ since $\bot$ is an abbreviation for $\emptyset^\vee$.
    $\setprop$ and $\bot$ have the same \htmodel s by the definition of strong equivalence, thus $\hti \not\models \setprop$.
    Now assume every \htinterp\ $\hti$ fails to satisfy $\setprop$.
    We know that every $\hti \not\models \bot$.
    Thus, $\setprop$ and $\bot$ have the same \htmodel s and $\setprop \equiv_s \bot$.
\end{proof}


\begin{fact}
    \label{fact:2}
    Let $\formprop(u,v)$ be an infinitary formula containing precomputed terms $u$ and $v$,
    and let $S_1$ and $S_2$ denote sets of precomputed terms of the same sort as $u$ and $v$, respectively.
    Then, the infinitary formulas
    \begin{gather}
        \label{eq:fact:2:1}
        \big \{ \{\formprop(u,v) \mid u \in S_1 \}^\vee \mid v \in S_2 \big\}^\vee
    \end{gather}
    and
    \begin{gather}
        \label{eq:fact:2:2}
        \big \{ \formprop(u,v) \mid \langle u, v \rangle \in S_1 \times S_2 \big\}^\vee
    \end{gather}
    are strongly equivalent.
    This fact can be generalized in a straightforward way to tuples of more than two terms. 
\end{fact}

\begin{proof}
    Recall that the condition of strong equivalence follows if we can establish that~\eqref{eq:fact:2:1} is satisfied by an \htinterp\ iff~\eqref{eq:fact:2:2} is satisfied by that same \htinterp.
    \\\noindent\emph{Left-to-right:} Take any \htinterp\ $\hti$ and assume that $\hti$ satisfies~\eqref{eq:fact:2:1}.
    By the definition of infinitary ht-satisfaction (bullet 3),  $\hti \models \{\formprop(u,y) \mid u \in S_1 \}^\vee$ for some $y$ belonging to $S_2$.
    Since $\hti \models \{\formprop(u,y) \mid u \in S_1 \}^\vee$, it follows from the definition of infinitary ht-satisfaction (bullet 3) that $\hti \models \formprop(x,y)$ for some $x$ belonging to $S_1$.
    Since $\hti \models \formprop(x,y)$ and the pair $\langle x, y \rangle$ belongs to $S_1 \times S_2$, it follows that $\hti$ satisfies~\eqref{eq:fact:2:2} (again, by the definition of infinitary ht-satisfaction).

    \emph{Right-to-left:} Take any \htinterp\ $\hti$ and assume that $\hti$ satisfies~\eqref{eq:fact:2:2}.
    It follows that $\hti \models \formprop(x,y)$ for some pair $\langle x, y \rangle \in S_1 \times S_2$.
    Thus, $\hti$ satisfies the formula $\formprop(u,y)$ for some $u$ (specifically, $u = x$, $u \in S_1$).
    It follows that $\hti$ also satisfies the formula $\{\formprop(u,y) \mid u \in S_1 \}^\vee$.
    Similarly, $\hti$ satisfies the formula $\{\formprop(u,v) \mid u \in S_1 \}^\vee$ for some $v$ (specifically, $v = y$, $v \in S_2$).
    From this, it follows that $\hti$ also satisfies~\eqref{eq:fact:2:1}.
\end{proof}


\begin{fact}
    \label{fact:3}
    Let $\intfolone$ be a standard interpretation, and let $\boldp,\boldq$ be a partition of predicate symbols in our signature (so that every comparison $\prec$ is in $\boldq$), and $t_1$ and $t_2$ are ground  terms. Then
        \begin{itemize}
            \item     $\grdstd(t_1\prec t_2)=\top$ when relation~$\prec$ holds for the pair $(t_1^\intfolone,t_2^\intfolone)$; and 
            \item $\grdstd(t_1\prec t_2)=\bot$, otherwise.    
        \end{itemize}  
\end{fact}

\begin{proof}
This fact follows from the definitions of grounding (condition~\ref{gr:condition:3}) and standard interpretations (condition~\ref{c.interp.seventh}).
\end{proof}


\begin{fact}
    \label{fact:4}
    \begin{enumerate}
        \item For any real numbers $n$ and $m$ such that $n < 0$ and $m = \abs{n}$, $-n = m$.
        \item For any real numbers $i$ and $j$, $\frac{\abs{i}}{\abs{j}} = \abs{\frac{i}{j}}$.
        \item For any real number $n$, $\ceil{ n } = - \floor{-n }$.
    \end{enumerate}  
\end{fact}

\begin{proof}
    The first two points follow directly from the definition of absolute value.
    For the last point, $\ceil{ n } = \ceil{ -(-n) }$, and $\ceil{ -(-n) }$ is the smallest integer $k$ such that $k \geq -(-n)$.
    Thus, $k \geq n$ and $-k \leq -n$.
    Therefore, $-k$ is the largest integer less than or equal to $-n$.
    By the definition of floor, $\floor{-n}= -k$.
    Since $\ceil{ -(-n) }$ is $k$ and $-k$ is $\floor{-n }$, it follows that $\ceil{ n } = - \floor{-n }$.    
\end{proof}


\begin{fact}{\cite[Corollary 1]{harlifpeval17}}
    \label{fact:5:harrison}
    Within an infinitary propositional formula $\formprop$, (infinitely many) parts of 
    $\formprop$ can be simultaneously replaced with strongly equivalent formulas without changing the set of stable models of $\formprop$.
\end{fact}


%% file: appendix-2.tex
\subsection{Proof of Proposition~\ref{prop:gr:val}}

\begin{lemma}
    \label{lem:std:abs.val}
    Let $\numeral{i},\numeral{j},\numeral{k}$ be numerals such that $k$ is the integer $\floor{\abs{i} / \abs{j} }$.
    Then,
    \begin{itemize}
        \item $\numeral{k}$ is $\numeral{round(i/j)}$ if $i \times j \geq 0$;
        \item $\numeral{-k}$ is $\numeral{round(i/j)}$ if $i \times j < 0$.
    \end{itemize}
\end{lemma}

\begin{proof}
    First note that the lemma condition implies that $j \neq 0$.
    \\\noindent Assume $i \times j \geq 0$. 
    It follows that $round(i/j) = \floor{i/j}$ and $i / j = \abs{i} / \abs{j}$, thus $\floor{i/ j } = \floor{\abs{i}/ \abs{j} }$.
    By the lemma conditions, $k = \floor{\abs{i}/ \abs{j} }$ and consequently, $k = \floor{i/ j }= round(i/j)$. It follows that
    $\numeral{k} = \numeral{round(i/j)}$.
    \\[5pt]
    \\\noindent Assume $i \times j < 0$.
    It follows that $round(i/j) = \ceil{i/j}$ and $\ceil{i/j} \leq 0$.
    \\\noindent\emph{Case 1}: 
    Assume $\ceil{ i / j } < 0$.
    Thus, $\abs{j} \leq \abs{i}$.
    Denote $\frac{i}{j}$ by $n$, and denote $\frac{\abs{i}}{\abs{j}}$ by $m$.
    It follows from points 1 and 2 of Fact~\ref{fact:4} that $-n = m$.
    Furthermore, it follows from point 3 of Fact~\ref{fact:4} that $\ceil{n} = - \floor{-n}$, and therefore $\ceil{n} = - \floor{m}$.
    Thus, $\ceil{ i / j } = - \floor{\abs{i} / \abs{j} } = -k$ (recall that $k$ is $\floor{\abs{i}/\abs{j} }$), which entails that $\numeral{-k}$ is $\numeral{\ceil{ i / j }}$.
    \\\noindent\emph{Case 2}: Assume $\ceil{ i / j } = 0$.
    Thus, $\abs{j} > \abs{i}$, and $0 \leq \abs{i}/\abs{j} < 1$. Consequently, $\floor{\abs{i}/\abs{j}}= 0$, recall that $k$ is $\floor{\abs{i}/\abs{j} }$.
    It follows that $k = 0$, and thus $-k = 0$. Consequently, $\numeral{-k}$ is $\numeral{\ceil{ i / j }}$.
\end{proof}


\begin{lemma}
    \label{lem:gr:val:f1}
    Let $\intfolone$ be a standard interpretation and let $F_1(IJK)$ be the formula
    $$K \times \lvert J \rvert \leq \lvert I \rvert < (K + \overline{1}) \times \lvert J \rvert$$
    Then $\grdstd(F_1(\numeral{i}\numeral{j}\numeral{k})) \equiv_s \top$ if 
    $\numeral{i},\numeral{j},\numeral{k}$ are numerals such that integer $k=\floor{\abs{i} / \abs{j} }$
    and $\grdstd(F_1(\numeral{i}\numeral{j}\numeral{k})) \equiv_s \bot$ otherwise.
\end{lemma}

\begin{proof}
    Recall that $F_1(IJK)$ is an abbreviation for the formula
    $$(K \times \lvert J \rvert \leq \lvert I \rvert) \wedge (\lvert I \rvert < (K + \overline{1}) \times \lvert J \rvert)$$
    thus, 
    \begin{align*}
        \grdstd(F_1(\numeral{i}\numeral{j}\numeral{k})) &= \grdstd(\numeral{k} \times \lvert \numeral{j} \rvert \leq \lvert \numeral{i} \rvert)) \wedge \grdstd(\lvert \numeral{i} \rvert < (\numeral{k} + \numeral{1}) \times \lvert \numeral{j} \rvert)
    \end{align*}
    Let $\numeral{n_1}$ denote the domain element $(\numeral{k} \times \lvert \numeral{j} \rvert)^\intfolone$, and let $\numeral{n_2}$ denote the domain element $(\abs{\numeral{i}})^\intfolone$.
    By the definition of grounding, $\grdstd(\numeral{k} \times \lvert \numeral{j} \rvert \leq \lvert \numeral{i} \rvert)) \equiv_s \top$ if the relation $\leq$ holds between $\numeral{n_1}$ and $\numeral{n_2}$, and $\grdstd(\numeral{k} \times \lvert \numeral{j} \rvert \leq \lvert \numeral{i} \rvert)) \equiv_s \bot$ otherwise.
    Since $\intfolone$ is a standard interpretation, $\numeral{n_1}$ and $\numeral{n_2}$ are the numerals $\numeral{k \times \abs{j}}$ and $\numeral{\abs{i}}$, respectively.

    Similarly, let $\numeral{n_3}$ denote the domain element $((\numeral{k} + \overline{1}) \times \abs{\numeral{j}})^\intfolone$.
    By the definition of grounding, $\grdstd(\lvert \numeral{i} \rvert < (\numeral{k} + \overline{1}) \times \lvert \numeral{j} \rvert) \equiv_s \top$ if the relation $<$ holds between $\numeral{n_2}$ and $\numeral{n_3}$, and $\grdstd(\lvert i \rvert < (k + \overline{1}) \times \lvert j \rvert) \equiv_s \bot$ otherwise.
    Since $\intfolone$ is a standard interpretation, $\numeral{n_2}$ and $\numeral{n_3}$ are the numerals $\numeral{\abs{i}}$ and $\numeral{(k + 1) \times \abs{j}}$, respectively.

    Since the total order on numerals mirrors the total order on integers, this means that we need to establish that $n_1 \leq n_2 < n_3$ for integers $n_1$, $n_2$, and $n_3$.
    We proceed by contradiction and assume that either $n_1 > n_2$ or $n_2 \geq n_3$.
    
    If $n_1 > n_2$, then $k \times \abs{j} > \abs{i}$. 
    Thus, $k > \abs{i} / \abs{j} \geq \floor{\abs{i} / \abs{j}}$. 
    But this contradicts the lemma's condition that $k$ is $\floor{\abs{i}/\abs{j}}$.

    If $n_2 > n_3$, then $\abs{i} > (k+1) \times \abs{j}$.
    Thus, $k+1 < \abs{i} / \abs{j}$. Let $l=\abs{i} / \abs{j} - \floor{\abs{i}/\abs{j}}$. Note how $l < 1$ and we can represent $\abs{i} / \abs{j}=\floor{\abs{i}/\abs{j}}+ l$.
    It is easy to see that $k+1 < \floor{\abs{i}/\abs{j}} + l$ and consequently, $k <  \floor{\abs{i}/\abs{j}} + (l - 1)$.
    Since $l-1 < 0$, $k <  \floor{\abs{i}/\abs{j}}$.
    This again contradicts the lemma's condition that $k = \floor{\abs{i}/\abs{j}}$.
    
    Finally, if $n_2 = n_3$, then $\abs{i} = (k+1) \times \abs{j}$ and consequently, $\abs{i} / \abs{j} = (k+1)$.
    Thus, \hbox{$\abs{i} / \abs{j} - 1 = k$}.
    Recall that $\abs{i} / \abs{j} = \floor{\abs{i}/\abs{j}} + l$, where $l$ is strictly less than $1$.
    Consequently, \hbox{$k =  \floor{\abs{i}/\abs{j}} + (l - 1)$},
    which contradicts that $k = \floor{\abs{i}/\abs{j}}$.
\end{proof}


\begin{lemma}
    \label{cor:gr:val:f1}
    Let $\intfolone$ be a standard interpretation and let $F_1(IJK)$ be the formula
    $$K \times \lvert J \rvert \leq \lvert I \rvert < (K + \overline{1}) \times \lvert J \rvert$$
    If $\numeral{i}$ and $\numeral{k}$ are any numerals and $\numeral{j}$ is $\overline{0}$, then $\grdstd(F_1(\numeral{i}\numeral{j}\numeral{k})) \equiv_s \bot$.
\end{lemma}

\begin{proof}
    Recall that 
    \begin{align*}
        \grdstd(F_1(\numeral{i}\numeral{j}\numeral{k})) &= \grdstd(\numeral{k} \times \lvert \numeral{j} \rvert \leq \lvert \numeral{i} \rvert)) \wedge \grdstd(\lvert \numeral{i} \rvert < (\numeral{k} + \numeral{1}) \times \lvert \numeral{j} \rvert)
    \end{align*}
    By the definition of grounding, $\grdstd(\lvert \numeral{i} \rvert < (\numeral{k} + \overline{1}) \times \lvert \numeral{j} \rvert) \equiv_s \bot$ if the relation $<$ does not hold between 
    $(\abs{\numeral{i}})^\intfolone$
    and 
    $((\numeral{k} + \overline{1}) \times \abs{\numeral{j}})^\intfolone$
  which is the case when $\abs{i}\geq (k+1)\times  \lvert {j} \rvert$.
Given that $\numeral{j}=\numeral{0}$, $(k+1)\times  \lvert {j} \rvert=0$. For any integer $i$, $\abs{i}\geq 0$ due to the definition of an absolute value.
\end{proof}


\begin{lemma}
    \label{lem:gr:val:f2}
    Let $\intfolone$ be a standard interpretation, let $r$ be a precomputed term, and let $F_2(IJKZ)$ be the formula
    $$(I \times J \geq \overline{0} \wedge Z = K) \vee (I \times J < \overline{0} \wedge Z = -K)$$
    Further, let $\numeral{i},\numeral{j},\numeral{k}$ be numerals such that $k$ is $\floor{\abs{i} / \abs{j} }$.
    Then $\grdstd(F_2(\numeral{i}\numeral{j}\numeral{k}r)) \equiv_s \top$ if $r$ is 
    the numeral $\numeral{round(i/j)}$ and $\grdstd(F_2(\numeral{i}\numeral{j}\numeral{k}r)) \equiv_s \bot$ otherwise.
\end{lemma}

\begin{proof}
    \begin{align*}
        \grdstd(F_2(\numeral{i}\numeral{j}\numeral{k}r))
        &=\grdstd\big((\numeral{i} \times \numeral{j} \geq \overline{0} \wedge r = \numeral{k}) \vee (\numeral{i} \times \numeral{j} < \overline{0} \wedge r = \numeral{-k})\big)\\
        &=\grdstd(\numeral{i} \times \numeral{j} \geq \overline{0} \wedge r = \numeral{k}) \vee gr(\numeral{i} \times \numeral{j} < \overline{0} \wedge r = \numeral{-k})\\
        &=\big(\grdstd(\numeral{i} \times \numeral{j} \geq \overline{0}) \wedge \grdstd(r = \numeral{k})\big) \vee 
        \big(\grdstd(\numeral{i} \times \numeral{j} < \overline{0}) \wedge \grdstd(r = \numeral{-k})\big)\\
    \end{align*}
    \noindent\emph{Case 1.} Assume $r$ is $\numeral{round(i/j)}$. 
    We need to show that $\grdstd(F_2(\numeral{i}\numeral{j}\numeral{k}r)) \equiv_s \top$.
    \\\noindent\emph{Case 1.1} Assume $i \times j \geq 0$. 
    By the definition of grounding, 
    $$
    \grdstd(F_2(\numeral{i}\numeral{j}\numeral{k}r)) = \big(\top \wedge \grdstd(r = \numeral{k})\big) \vee \big(\bot \vee \grdstd(r = \numeral{-k})\big)
    $$
    which entails (by Fact~\ref{fact:1}) that $\grdstd(F_2(\numeral{i}\numeral{j}\numeral{k}r)) \equiv_s \grdstd(r = \numeral{k})$.
    $\grdstd(r = \numeral{k})$ is $\top$ iff $r^\intfolone = \numeral{k}^\intfolone$ iff $\numeral{round(i/j)}^\intfolone = \numeral{k}^\intfolone$. 
    This last condition follows immediately from Lemma~\ref{lem:std:abs.val}.
    \\\noindent\emph{Case 1.2} Assume $i \times j < 0$. 
    By the definition of grounding, 
    $$
    \grdstd(F_2(\numeral{i}\numeral{j}\numeral{k}r)) = \big(\bot \wedge \grdstd(r = \numeral{k})\big) \vee \big(\top \vee \grdstd(r = \numeral{-k})\big)
    $$
    which entails (by Fact~\ref{fact:1}) that $\grdstd(F_2(\numeral{i}\numeral{j}\numeral{k}r)) \equiv_s \grdstd(r = \numeral{-k})$.
    $\grdstd(r = \numeral{-k})$ is $\top$ iff $r^\intfolone = \numeral{-k}^\intfolone$ iff $\numeral{round(i/j)}^\intfolone = \numeral{-k}^\intfolone$. 
    This last condition follows immediately from Lemma~\ref{lem:std:abs.val}.
    \\[5pt]
    \\\noindent\emph{Case 2.} Assume $r$ is not $\numeral{round(i/j)}$.
    We need to show that $gr(F_2(\numeral{i}\numeral{j}\numeral{k}r)) \equiv_s \bot$.
    \\\noindent\emph{Case 2.1} Assume $i \times j \geq 0$. 
    Lemma~\ref{lem:std:abs.val} implies that $\numeral{k}$ is $\numeral{round(i/j)}$.
    As in Case 1.1, we derive that $\grdstd(F_2(\numeral{i}\numeral{j}\numeral{k}r)) \equiv_s \grdstd(r = \numeral{k})$. Thus,
    $\grdstd(r = \numeral{k})$ is $\bot$ iff $r^\intfolone \neq \numeral{k}^\intfolone$.
    As $r$ is not $\numeral{round(i/j)}$, $r^\intfolone\neq \numeral{round(i/j)}$. It follows that 
    $r^\intfolone\neq\numeral{k}^\intfolone$.    
    \\\noindent\emph{Case 2.2} Assume $i \times j < 0$. 
    Lemma~\ref{lem:std:abs.val} implies that $\numeral{-k}$ is $\numeral{round(i/j)}$.
    As in Case 1.2, we derive that  $\grdstd(F_2(\numeral{i}\numeral{j}\numeral{k}r)) \equiv_s \grdstd(r = \numeral{-k})$.
    Thus, $\grdstd(r = \numeral{-k})$ is $\bot$ iff 
    $r^\intfolone \neq \numeral{-k}^\intfolone$ or, in other words, $r^\intfolone \neq \numeral{-k}$.
    Just as in Case 1.2, $r$ is not $\numeral{round(i/j)}$, but $\numeral{-k}$ is.
\end{proof}


\begin{lemma}
    \label{lem:gr:val:f3}
    Let $\intfolone$ be a standard interpretation, let $r$ be a precomputed term, and let $F_3(IJKZ)$ be the formula
    $$(I \times J \geq \overline{0} \wedge Z = I - K \times J) \vee (I \times J \leq \overline{0} \wedge Z = I + K \times J)$$
    Further, let $\numeral{i},\numeral{j},\numeral{k}$ be numerals such that $k$ is $\floor{\abs{i} / \abs{j} }$.
    Then $\grdstd(F_3(\numeral{i}\numeral{j}\numeral{k}r)) \equiv_s \top$ if $r$ is 
    the numeral $\numeral{i - j \times round(i/j)}$ and $\grdstd(F_3(\numeral{i}\numeral{j}\numeral{k}r)) \equiv_s \bot$ otherwise.
\end{lemma}

\begin{proof}
    \begin{align*}
        \grdstd(F_3(\numeral{i}\numeral{j}\numeral{k}r)) &= \grdstd\big((\numeral{i} \times \numeral{j} \geq \overline{0} \wedge r = \numeral{i} - \numeral{k} \times \numeral{j}) \vee (\numeral{i} \times \numeral{j} < \overline{0} \wedge r = \numeral{i} + \numeral{k} \times \numeral{j})\big)\\
        &= \big(\grdstd(\numeral{i} \times \numeral{j} \geq \overline{0}) \wedge \grdstd(r = \numeral{i} - \numeral{k} \times \numeral{j})\big) \vee \big(\grdstd(\numeral{i} \times \numeral{j} < \overline{0}) \wedge \grdstd(r = \numeral{i} + \numeral{k} \times \numeral{j})\big)
    \end{align*}
    \emph{Case 1.} Assume $i \times j \geq 0$.
    It follows that $\grdstd(F_3(\numeral{i}\numeral{j}\numeral{k}r)) \equiv_s \grdstd(r = \numeral{i} - \numeral{k} \times \numeral{j})$.
    From Lemma~\ref{lem:std:abs.val} it follows that $\numeral{k}$ is $\numeral{round(i/j)}$.
    Thus, by the definition of grounding it follows that \hbox{$\grdstd(F_3(\numeral{i}\numeral{j}\numeral{k}r)) \equiv_s \top$} if $r^\intfolone = (\numeral{i} - \numeral{round(i/j)} \times \numeral{j})^\intfolone$, and $\grdstd(F_3(\numeral{i}\numeral{j}\numeral{k}r)) \equiv_s \bot$ otherwise.
    Consequently, \hbox{$\grdstd(F_3(\numeral{i}\numeral{j}\numeral{k}r)) \equiv_s \top$} if $r$ is 
    the numeral $\numeral{i - j \times round(i/j)}$ and $\grdstd(F_3(\numeral{i}\numeral{j}\numeral{k}r)) \equiv_s \bot$ otherwise.
    
    \emph{Case 2.} Assume $i \times j < 0$.
    It follows that $\grdstd(F_3(\numeral{i}\numeral{j}\numeral{k}r)) \equiv_s \grdstd(r = \numeral{i} + \numeral{k} \times \numeral{j})$.
    From Lemma~\ref{lem:std:abs.val} it follows that $\numeral{-k}$ is $\numeral{round(i/j)}$.
    Thus, $\numeral{-k}^\intfolone = \numeral{round(i/j)}^\intfolone$.
    By properties over numerals, $\numeral{k} = \numeral{-round(i/j)}$.
    Therefore, it follows from the definition of grounding that $\grdstd(F_3(\numeral{i}\numeral{j}\numeral{k}r)) \equiv_s \top$ if $r^\intfolone = (\numeral{i} + \numeral{-round(i/j)} \times \numeral{j})^\intfolone$ and $\grdstd(F_3(\numeral{i}\numeral{j}\numeral{k}r)) \equiv_s \bot$ otherwise.
    Consequently, \hbox{$\grdstd(F_3(\numeral{i}\numeral{j}\numeral{k}r)) \equiv_s \top$} if $r$ is 
    the numeral $\numeral{i - j \times round(i/j)}$ and $\grdstd(F_3(\numeral{i}\numeral{j}\numeral{k}r)) \equiv_s \bot$ otherwise.
\end{proof}


\paragraph{Proof of Proposition~\ref{prop:gr:val}}

\begin{proof}
    We proceed by structural induction across the distinct cases of forms of terms.
    For the duration of this proof, we write $gr$ to denote  $\grdstd$.
    Furthermore, we do not distinguish numerals from integers in this proof using these terms interchangeably, due to the fact that numerals and integers are in one to one correspondence.
    Thus, we will abuse the notation and drop the overline symbol denoting a numeral.
    \\[5pt]
    \\\noindent\emph{Case 1}: Term $t$ is a numeral, symbolic constant, $\mathit{inf}$, or $\mathit{sup}$. 
    Then,  $val_t(r)$ is $r = t$ and $[t] = \{t\}$.
    By definition, $gr(r = t) = \top$ if $r^\intfolone = t^\intfolone$ and $\bot$ otherwise.
    Given that $I$ is a standard interpretation, $r^I = t^I$ if and only if $r \in [t]$.
    \\[5pt]
    
    In the remainder of the proof, we say that a term $t$ has the induction hypothesis property when 
     the formula $gr(val_{t}(r))$ is strongly equivalent to $\top$ if $r \in [t]$ and to $\bot$ otherwise.
    \\\noindent\emph{Case 2}: Term $t$ is $\lvert t_1 \rvert$, where $t_1$ has the induction hypothesis property.    
    By definition, $val_t(r)$ is $\exists I (val_{t_1}(I) \wedge r = \vert I \rvert))$. Then,
 \begin{align*}
        gr(val_t(r)) &=gr(\exists I (val_{t_1}(I) \wedge r = \vert I \rvert))\\ 
        &= \{gr(val_{t_1}(i)) \wedge gr(r = \lvert i \rvert) \mid i \in |\I|^{\sortint}\}^\vee\\
        &= \big\{
        (gr(val_{t_1}(j)) \wedge gr(r = \lvert j \rvert)),\\
        &~~~~~~~(gr(val_{t_1}(k))\wedge gr(r = \lvert k \rvert)) \mid j,k \in |\I|^{\sortint},j\in[t_1],k\not\in[t_1]\big\}^\vee
    \end{align*}
    
    By the fact that $t_1$ has the inductive property, it follows 
    that
    \begin{align*}
        gr(val_t(r)) 
        &\equiv_s\big\{
        (\top \wedge gr(r = \lvert j \rvert)),\\
        &~~~~~~~~( \bot\wedge gr(r = \lvert k \rvert)) \mid j,k \in |\I|^{\sortint},j\in[t_1],k\not\in[t_1]\big\}^\vee
     \end{align*}   
     In turn, this formula is strongly equivalent the following formula, by  Fact~\ref{fact:1} (conditions \ref{l:1:1}, \ref{l:1:2}) and Fact~\ref{fact:5:harrison}).
    \begin{align}
        \label{eq:prop1:case2}
        \big\{
        gr(r = \lvert j \rvert)
         \mid j \in |\I|^{\sortint},j\in[t_1]\big\}^\vee.
    \end{align}   
    \\[5pt]
    \\\noindent Case 2.1: $r \in [t]$, where $t = \lvert t_1 \rvert$. 
    
    Our goal is to show that~\eqref{eq:prop1:case2} is strongly equivalent to $\top$, which holds
    \\iff, for an arbitrary \htinterp $\hti$, $\hti \models$~\eqref{eq:prop1:case2}\hfill{(Fact~\ref{fact:1}, point 3)}
    \\iff, for an arbitrary \htinterp $\hti$, \\
    \phantom{iff, }  $\hti \models gr(r = \lvert j \rvert)$ for some $j \in |\I|^{\sortint},j\in[t_1]$
    \\iff, $gr(r = \lvert j \rvert)=\top$ for some $j \in |\I|^{\sortint},j\in[t_1]$
    \\iff, for some numeral $j \in [t_1]$, $r^\intfolone = \lvert j \rvert^\intfolone$\hfill{(Definition of grounding)}
    \\iff, for some numeral $j \in [t_1]$, $r = \lvert j \rvert$\hfill{(Conditions~\ref{c.interp.third} and~\ref{c.interp.fifth} of standard interpretations)}
    
    Now, by our assumption that $r \in [t]$ it follows that $r \in \{\lvert j \rvert \mid j \in [t_1]\}$.
    Hence, there exists a $j \in [t_1]$ such that $r = \lvert j \rvert$, and our goal is satisfied.
    \\[5pt]
    \\\noindent Case 2.2: $r \not\in [t]$, where $t = \lvert t_1 \rvert$. 
    
    Our goal is to show that~\eqref{eq:prop1:case2} is strongly equivalent to $\bot$, which holds
    \\iff, for an arbitrary \htinterp $\hti$, $\hti \not\models$~\eqref{eq:prop1:case2}\hfill{(Fact~\ref{fact:1}, point 4)}
    \\iff, for an arbitrary \htinterp $\hti$, $\hti \not\models H$ for any $H$ in~\eqref{eq:prop1:case2}\hfill{(Definition)}
    \\iff, for an arbitrary \htinterp $\hti$, $\hti \not\models gr(r = \lvert j \rvert)$ for any $j \in |\I|^{\sortint},j\in[t_1]$
     \\iff, $gr(r = \lvert j \rvert)=\bot$ for any $j \in |\I|^{\sortint},j\in[t_1]$
    \\iff, for any numeral $j \in [t_1]$, $r^\intfolone \neq \lvert j \rvert^\intfolone$\hfill{(Definition of grounding)}
    \\iff, for any numeral $j \in [t_1]$, $r \neq \lvert j \rvert$\hfill{(Conditions~\ref{c.interp.third} and~\ref{c.interp.fifth} of standard interpretations)}
    
    Now, by our assumption that $r \not\in [t]$ it follows that $r \not\in \{\lvert j \rvert \mid j \in [t_1]\}$. Consequently, for any  $j$ in $[t_1]$,  $r\neq|j|$.
    \\[5pt]
    \\\noindent\emph{Case 3}: Term $t$ is $t_1 \circ t_2$, where $\circ$ is one of the operation names $+\;$, $ -\;$,  $ \times$, and $t_1$ and $t_2$ have the induction hypothesis property.
    Then, $val_t(r)$ is $\exists I J (r = I \circ J \wedge val_{t_1}(I) \wedge val_{t_2}(J))$, and
    \begin{align*}
        gr(val_t(r)) &= gr(\exists I J (r = I \circ J \wedge val_{t_1}(I) \wedge val_{t_2}(J)))\\
        &= \{ gr(\exists J (r = i \circ J \wedge val_{t_1}(i) \wedge val_{t_2}(J))) \mid i \in \universe{\intfolone}{\sortint} \}^\vee\\
        &= \big\{ \{ gr(r = i \circ j) \wedge gr(val_{t_1}(i)) \wedge gr(val_{t_2}(j)) \mid i \in \universe{\intfolone}{\sortint} \}^\vee \mid j \in \universe{\intfolone}{\sortint} \big\}^\vee\\
        &= \big\{ \{(gr(r = i_1 \circ j) \wedge gr(val_{t_1}(i_1)) \wedge gr(val_{t_2}(j))),\\
        &~~~~~~~~~~(gr(r = i_2 \circ j) \wedge gr(val_{t_1}(i_2)) \wedge gr(val_{t_2}(j))) \mid i_1, i_2 \in \universe{\intfolone}{\sortint}, i_1 \in [t_1], i_2 \not\in [t_1] \}^\vee \mid j \in \universe{\intfolone}{\sortint} \big\}^\vee
    \end{align*}
    By the fact that $t_1$ has the inductive property, Fact~\ref{fact:1} (conditions \ref{l:1:1}, \ref{l:1:2}), and Fact~\ref{fact:5:harrison},
    it follows that
    \begin{align*}
        gr(val_t(r))
        &\equiv_s \big\{ \{gr(r = i_1 \circ j) \wedge gr(val_{t_2}(j)) \mid i_1 \in \universe{\intfolone}{\sortint}, i_1 \in [t_1] \}^\vee \mid j \in \universe{\intfolone}{\sortint} \big\}^\vee
    \end{align*}
    In turn, this formula can be rewritten as
    \begin{align*}
    \Big\{ \big\{
        &\big(gr(r = i_1 \circ j_1) \wedge gr(val_{t_2}(j_1))\big), \big(gr(r = i_1 \circ j_2) \wedge gr(val_{t_2}(j_2))\big) \bmid  i_1 \in \universe{\intfolone}{\sortint}, i_1 \in [t_1]\big\}^\vee \\
        &~~~~~~~~~~~~~~~~~~~~~~~~~~~~~~~~~~~~~~~~\bmid  j_1,  j_2 \in \universe{\intfolone}{\sortint},
        j_1 \in [t_2], j_2 \not\in [t_2]\Big\}^\vee 
    \end{align*}
    By the fact that $t_2$ has the inductive property,   Fact~\ref{fact:1} (conditions \ref{l:1:1}, \ref{l:1:2}), and Fact~\ref{fact:5:harrison},
    it follows that
    \begin{align*}
        gr(val_t(r)) 
        &\equiv_s \big\{ \{
        gr(r = i_1 \circ j_1) \mid  i_1 \in \universe{\intfolone}{\sortint}, i_1 \in [t_1]\}^\vee \mid  j_1 \in \universe{\intfolone}{\sortint}, j_1 \in [t_2] \big\}^\vee 
    \end{align*}
    By Fact~\ref{fact:2}, $gr(val_t(r))$ is strongly equivalent to (where we rename $i_1$ by $i$ and $j_1$ by $j$ for readability)
    \begin{gather}
        \label{eq:prop1:case3}
        \{gr(r = i \circ j) \mid i, j \in \universe{\intfolone}{\sortint}, i \in [t_1], j \in [t_2]\}^\vee
    \end{gather}
    \\\noindent Case 3.1: $r \in [t]$, where $t = t_1 \circ t_2$. 
    
    Our goal is to show that~\eqref{eq:prop1:case3} is strongly equivalent to $\top$, which holds
    \\iff, for an arbitrary \htinterp $\hti$, $\hti \models gr(r = i \circ j)$ for some 
    \hbox{$i, j \in \universe{\intfolone}{\sortint}$}, $i \in [t_1], j \in [t_2]$ 
    \\iff, for some pair of numerals $i, j$ such that $i \in [t_1], j \in [t_2]$, $gr(r = i \circ j) = \top$
    \\iff, for some pair of numerals $i, j$ such that $i \in [t_1], j \in [t_2]$, $r^\intfolone = (i \circ j)^\intfolone$\hfill{(Definition of grounding and satisfaction)}
    \\iff, for some pair of numerals $i, j$ such that $i \in [t_1], j \in [t_2]$, $r = i \circ j$\hfill{(Conditions~\ref{c.interp.third} and~\ref{c.interp.fourth} of standard interpretations)}\\
    Now, by our assumption that $r \in [t]$ it follows that $r \in \{{i \circ j} \mid i \in [t_1], j \in [t_2]\}$. 
    Hence, there exists  $i \in [t_1]$ and  $j \in [t_2]$  such that $r = i \circ j$, and the goal immediately follows.
    \\\noindent Case 3.2: $r \not\in [t]$, where $t = t_1 \circ t_2$.
    Our goal is to show that~\eqref{eq:prop1:case3} is strongly equivalent to $\bot$, which holds
    \\iff, for an arbitrary \htinterp $\hti$, $\hti \not\models gr(r = i \circ j)$ for any 
    \hbox{$i, j \in \universe{\intfolone}{\sortint}$}, $i \in [t_1], j \in [t_2]$ 
    \\iff, for any pair of numerals $i, j$ such that $i \in [t_1], j \in [t_2]$, $gr(r = i \circ j) = \bot$
    \\iff, for any pair of numerals $i, j$ such that $i \in [t_1], j \in [t_2]$, $r^\intfolone \neq (i \circ j)^\intfolone$\hfill{(Definition of grounding and satisfaction)}
    \\iff, for any pair of numerals $i, j$ such that $i \in [t_1], j \in [t_2]$, $r \neq i \circ j$\hfill{(Conditions~\ref{c.interp.third} and~\ref{c.interp.fourth} of standard interpretations)}\\
    Now, by our assumption that $r \not\in [t]$ it follows that $r \not\in \{{i \circ j} \mid i \in [t_1], j \in [t_2]\}$. 
    Consequently, for any $i \in [t_1]$ and  $j \in [t_2]$ it follows that $r \neq i \circ j$.
    \\[5pt]
    \\\noindent\emph{Case 4}: Term $t$ is $(t_1..t_2)$ where $t_1$ and $t_2$ have the induction hypothesis property.
    Then, $val_t(r)$ is $\exists I J K (r = K \wedge I \leq K \wedge K \leq J \wedge val_{t_1}(I) \wedge val_{t_2}(J))$, and
    \begin{align*}
        &gr(val_t(r)) = \\
        &~~~\big\{ \big\{ \{ gr(r = k) \wedge gr(i \leq k) \wedge gr(k \leq j) \wedge gr(val_{t_1}(i)) \wedge gr(val_{t_2}(j)) \mid k \in \universe{\intfolone}{\sortint} \}^\vee\\ 
        &~~~~~\mid j \in \universe{\intfolone}{\sortint} \big\}^\vee\\ 
        &~~~~~~~~~\mid i \in \universe{\intfolone}{\sortint}\big\}^\vee
    \end{align*}
    Using the same style of transformations and line of reasoning as in the prior cases and the fact that~$t_1$ and~$t_2$ have the inductive property, we are 
    able to show that $gr(val_t(r))$ is strongly equivalent to 
    $$
    \big\{gr(r = k) \wedge gr(i \leq k) \wedge gr(k \leq j) \bmid  i, j, k \in \universe{\intfolone}{\sortint}, i \in [t_1], j \in [t_2]\big\}^\vee
    $$
    Every element in this infinitary formula belongs to one of the following sets
    \begin{gather}
        \label{eq:prop1:case4:in}
        \{gr(r = k) \wedge gr(i \leq k) \wedge gr(k \leq j) \mid i, j, k \in \universe{\intfolone}{\sortint}, i \in [t_1], j \in [t_2], i \leq k \leq j\}
    \end{gather}
    and 
    \begin{gather}
        \label{eq:prop1:case4:notin}
        \{gr(r = k) \wedge gr(i \leq k) \wedge gr(k \leq j) \mid i, j, k \in \universe{\intfolone}{\sortint}, i \in [t_1], j \in [t_2], \neg(i \leq k \leq j)\}
    \end{gather}

By Fact~\ref{fact:3} and the definition of HT-satisfaction,
    $$
    gr(i \leq k) \wedge gr(k \leq j) \equiv_s \top
    $$
    iff $i \leq k \leq j$.
    Similarly,
    $$
    gr(i \leq k) \wedge gr(k \leq j) \equiv_s \bot
    $$
    iff $\neg(i \leq k \leq j)$.

    Thus, every formula in~\eqref{eq:prop1:case4:in} is strongly equivalent to $gr(r = k)$, and every formula in~\eqref{eq:prop1:case4:notin} is strongly equivalent to $\bot$.
    It follows that 
    \begin{gather}
        \label{eq:prop1:case4:final}
        gr(val_t(r)) \equiv_s \{gr(r = k) \mid i, j, k \in \universe{\intfolone}{\sortint}, i \in [t_1], j \in [t_2], i \leq k \leq j\}^\vee
    \end{gather}

    The following chain of reasoning mirrors the proof by cases ($r \in [t]$, $r \not\in [t]$) detailed earlier.
    First note that by the definition of the values of $t$ ($t$ is $t_1..t_2$), $r \in [t]$ if $r$ belongs to the set of numerals
    \begin{gather}
        \label{eq:val:set:interval}
        \{m \mid n_1 \in [t_1], n_2 \in [t_2], n_1 \leq m \leq n_2\}
    \end{gather}
    and $r \not\in [t]$ if $r$ does not belong to~\eqref{eq:val:set:interval}.
    
    It follows from~\eqref{eq:prop1:case4:final} that $gr(val_t(r)) \equiv_s \top$ when $gr(r = k)=\top$ for some triple of numerals $i,j,k$ such that $i \in [t_1], j \in [t_2]$, $i \leq k \leq j$, and $gr(val_t(r)) \equiv_s \bot$ otherwise.
    The condition $gr(r = k)=\top$ for some triple of numerals $i,j,k$ such that $i \in [t_1], j \in [t_2]$, $i \leq k \leq j$ holds when $r^\intfolone = k^\intfolone$.
    Therefore, $gr(val_t(r)) \equiv_s \top$ when $r$ belongs to~\eqref{eq:val:set:interval} (take $k$ to be $r$) and $gr(val_t(r)) \equiv_s \bot$, otherwise.
    Consequently, $gr(val_t(r)) \equiv_s \top$ if $r \in [t]$, and $gr(val_t(r)) \equiv_s \bot$ otherwise.
    \\[5pt]
    \\\noindent\emph{Case 5}: Term $t$ is $t_1/t_2$, where $t_1$ and $t_2$ have the induction hypothesis property.
    Then, $val_t(r)$ is $\exists I J K(val_{t_1}(I) \wedge val_{t_2}(J) \wedge F_1(IJK) \wedge F_2(IJKr)$, and
    \begin{align*}
        &gr(val_t(r)) = \\
        &~~~\big\{ \big\{ \{ gr(val_{t_1}(i)) \wedge gr(val_{t_2}(j)) \wedge gr(F_1(ijk)) \wedge gr(F_2(ijkr)) \bmid k \in \universe{\intfolone}{\sortint} \}^\vee\\ 
        &~~~~~\bmid j \in \universe{\intfolone}{\sortint} \big\}^\vee\\ 
        &~~~~~~~~~\bmid i \in \universe{\intfolone}{\sortint}\big\}^\vee
    \end{align*}
    By Fact~\ref{fact:2},
    \begin{align*}
        gr(val_t(r)) &\equiv_s  \big\{ gr(val_{t_1}(i)) \wedge gr(val_{t_2}(j)) \wedge gr(F_1(ijk)) \wedge gr(F_2(ijkr)) \bmid i,j,k \in  \universe{\intfolone}{\sortint}\big\}^\vee
    \end{align*}
    Since $t_1$ and $t_2$ have the inductive property, we can use analogous reasoning to previous cases to conclude that
    \begin{gather}
        gr(val_t(r)) \equiv_s  \big\{ gr(F_1(ijk)) \wedge gr(F_2(ijkr)) \bmid i,j,k \in  \universe{\intfolone}{\sortint}, i \in [t_1], j \in [t_2] \big\}^\vee
        \label{eq:prop1:case5}
    \end{gather}

    We now illustrate that $gr(val_t(r)) \equiv_s \top$, when $r \in [t]$. 
      
    1.     It follows from~\eqref{eq:prop1:case5} that $gr(val_t(r)) \equiv_s \top$ iff $gr(F_1(ijk)) \equiv_s \top$ and $gr(F_2(ijkr)) \equiv_s \top$ for some triple  $i,j,k$ of numerals such that $i \in [t_1]$ and $j \in [t_2]$.
    
    2. From Lemma~\ref{lem:gr:val:f1} it follows that $gr(F_1(ijk)) \equiv_s \top$ when $k=\floor{\abs{i} / \abs{j} }$.
        
    3. By 1 and 2,    $gr(val_t(r)) \equiv_s \top$ if $gr(F_2(ijkr)) \equiv_s \top$ for some  $i,j,k$ triple of numerals  such that $i \in [t_1]$, $j \in [t_2]$, and $k=\floor{\abs{i} / \abs{j} }$. 
       
    4.    From Lemma~\ref{lem:gr:val:f2} it follows that $gr(F_2(ijkr)) \equiv_s \top$ when $r$ is 
        the numeral $round(i/j)$ for some triple  $i,j,k$ of numerals  such that  $k=\floor{\abs{i} / \abs{j} }$.
        
    5. By 3 and 4,  $gr(val_t(r)) \equiv_s \top$ when there is some triple  $i,j,k$ of numerals  such that $i \in [t_1]$, $j \in [t_2]$,  and $r$ is the numeral $round(i/j)$.

    6. By the definitions of the values of $t$, when $r\in[t]$, then there exist numerals $i,j$ such that $i \in [t_1]$, $j \in [t_2]$, $j \neq 0$, and $r$ is the numeral $round(i/j)$. 
      
    7. By 5 and 6,  $gr(val_t(r)) \equiv_s \top$  when $r \in [t]$.
    
    We now illustrate that $gr(val_t(r)) \equiv_s \bot$, when  $r \not \in [t]$.
    
    1. It follows from~\eqref{eq:prop1:case5} that $gr(val_t(r)) \equiv_s \bot$ when for all triples $i,j,k$ of numerals such that $i \in [t_1]$ and $j \in [t_2]$, either $gr(F_1(ijk)) \equiv_s \bot$ or $gr(F_2(ijkr)) \equiv_s \bot$.
    
    Consider an arbitrary triple $i,j,k$ of numerals such that $i \in [t_1]$ and $j \in [t_2]$.
    
    2.  From Lemma~\ref{lem:gr:val:f1} it follows that $gr(F_1(ijk)) \equiv_s \bot$ when $k \neq \floor{\abs{i} / \abs{j} }$.
    
    3. From Lemma~\ref{cor:gr:val:f1} it follows that $gr(F_1(ijk)) \equiv_s \bot$ when $j=0$.
    
    From now on, we consider the case when $j\neq 0$ and $k=\floor{\abs{i} / \abs{j} }$.
    
    4. Under these assumptions it is sufficient to show that $gr(F_2(ijkr)) \equiv_s \bot$ to conclude that
    $gr(val_t(r)) \equiv_s \bot$ (see 1,2,3).
    
    5. By the definition of the value of $t$, when $r\not\in[t]$, then there do not exist numerals $i,j$ such that $i \in [t_1]$, $j \in [t_2]$, $j \neq 0$, and $r$ is the numeral $round(i/j)$. 
    
    6. From Lemma~\ref{lem:gr:val:f2}, 
    we derive that 
    $gr(F_2(ijkr)) \equiv_s \bot$ in case 
    $r$ is not the numeral $round(i/j)$.
    
    7. By 4, { 5} and 6, we conclude that $gr(val_t(r)) \equiv_s \bot$ when $r \not\in [t]$.
    \\[5pt]
    \\\noindent\emph{Case 6}:
   
    Term $t$ is $t_1 \setminus t_2$, where $t_1$ and $t_2$ have the induction hypothesis property.
    Then, $val_t(r)$ is $\exists I J K\big(val_{t_1}(I) \wedge val_{t_2}(J) \wedge F_1(IJK) \wedge F_3(IJKr)\big)$.
    By similar reasoning to that employed in Case 5, we  conclude that
    \begin{gather}
        \label{eq:prop1:case6}
        gr(val_t(r)) \equiv_s  \big\{ gr\big(F_1(ijk)\big) \wedge gr\big(F_3(ijkr)\big) \bmid i,j,k \in  \universe{\intfolone}{\sortint}, i \in [t_1], j \in [t_2] \big\}^\vee
    \end{gather}
    We now illustrate that $gr(val_t(r)) \equiv_s \top$, when $r \in [t]$. 
    
    1. It follows from~\eqref{eq:prop1:case6} that $gr(val_t(r)) \equiv_s \top$ if $gr(F_1(ijk)) \equiv_s \top$ and \hbox{$gr(F_3(ijkr)) \equiv_s \top$} for some triple of numerals $i,j,k$ such that $i \in [t_1]$ and $j \in [t_2]$.

    2. From Lemma~\ref{lem:gr:val:f1} it follows that $gr(F_1(ijk)) \equiv_s \top$ when $k=\floor{\abs{i} / \abs{j} }$.  
    Consequently, 
    $gr(val_t(r)) \equiv_s \top$ if $gr(F_3(ijkr)) \equiv_s \top$ for some triple of numerals $i,j,k$ such that $i \in [t_1]$, $j \in [t_2]$, and $k=\floor{\abs{i} / \abs{j} }$. 
    
    3. From Lemma~\ref{lem:gr:val:f3} it follows that $gr(F_3(ijkr)) \equiv_s \top$ if $r$ is 
    the numeral $i - j \times round(i/j)$ for $i,j,k$.
    
    4. By 2 and 3, $gr(val_t(r)) \equiv_s \top$ if there exists a triple of numerals $i,j,k$ such that $i \in [t_1]$, $j \in [t_2]$,  $k=\floor{\abs{i} / \abs{j} }$, and $r$ is $i - j \times round(i/j)$.
    
    5. By the definition of the value of $t$, when $r\in[t]$, then there exist numerals $i,j$ such that $i \in [t_1]$, $j \in [t_2]$, $j \neq 0$, and $r$ is the numeral $i - j \cdot round(i/j)$. 

    6. By 4 and 5,  $gr(val_t(r)) \equiv_s \top$ when $r \in [t]$.

\noindent
    We now illustrate that $gr(val_t(r)) \equiv_s \bot$, when  $r \not \in [t]$.

    1. It follows from~\eqref{eq:prop1:case6} that $gr(val_t(r)) \equiv_s \bot$ when for all triples of numerals $i,j,k$ such that $i \in [t_1]$ and $j \in [t_2]$, either $gr(F_1(ijk)) \equiv_s \bot$ or $gr(F_3(ijkr)) \equiv_s \bot$.

    Consider an arbitrary triple $i,j,k$ of numerals such that $i \in [t_1]$ and $j \in [t_2]$.

    2. Recall from Case 5 that $gr(F_1(ijk)) \equiv_s \bot$ when $k \neq \floor{\abs{i} / \abs{j} }$ or $j=0$.

    3. By 1 and 2, $gr(val_t(r)) \equiv_s \bot$ when $gr(F_3(ijkr)) \equiv_s \bot$ in case $k = \floor{\abs{i} / \abs{j} }$ and $j \neq 0$.
    
    4. By the definition of the value of $t$, when $r\not\in[t]$, then there do not exist numerals $i,j$ such that $i \in [t_1]$, $j \in [t_2]$, $j \neq 0$, and $r$ is the numeral $i - j \cdot round(i/j)$. 


    5. By 3, 4, and Lemma~\ref{lem:gr:val:f3}, we conclude that $gr(val_t(r)) \equiv_s \bot$  when $r \not\in [t]$.
\end{proof}


\paragraph{Proof of Corollary~\ref{prop:gr:val:corr}}

The following corollary is a consequence of Proposition~\ref{prop:gr:val} and is essentially a restatement of its claim for the case of a tuple of variables in place of a single variable.
Recall that for a tuple of terms $t_1,\dots,t_k$, abbreviated as $\bold{t}$, and a tuple of variables $V_1,\dots,V_k$, abbreviated as $\bold{V}$, we use $val_{\bold{t}}(\bold{V})$ to denote the formula 
$$val_{t_1}(V_1) \wedge \dots \wedge val_{t_k}(V_{t_k}).$$
\begin{corollary}
    Let $\intfolone$ be a standard interpretation, and let $\boldp$ be the standard partition.
    Then, for any $k$-tuple of program variables $\boldZ$, $k$-tuple of ground program terms $\boldt$, and $k$-tuple of precomputed terms $\boldr$, the formula $\grdstd(val_\boldt(\boldZ)_\boldr^\boldZ)$ is strongly equivalent to $\top$ if $\langle r_1, \dots r_k \rangle \in [t_1,\dots,t_k]$ and to $\bot$ otherwise.
    \label{prop:gr:val:corr}
\end{corollary}

\begin{proof}
    As before, we write $val_\boldt(\boldr)$ instead of $val_\boldt(\boldZ)_\boldr^\boldZ$ to denote the formula $val_\boldt(\boldZ)$ where $r_1$ replaces $Z_1$, and so forth.
    Thus,
    $$
    \grdstd(val_\boldt(\boldr)) = \grdstd(val_{t_1}(r_1)) \wedge \dots \wedge \grdstd(val_{t_k}(r_k))
    $$
    \\\noindent Case 1: $\langle r_1, \dots r_k \rangle \in [t_1,\dots,t_k]$.
    By definition, $r_i \in [t_i]$ ($1 \leq i \leq k$), thus by Proposition~\ref{prop:gr:val} each formula $\grdstd(val_{t_i}(r_i)) \equiv_s \top$.
    Since every conjunctive term in $\grdstd(val_\boldt(\boldr))$ is strongly equivalent to $\top$, it follows that $\grdstd(val_\boldt(\boldr))$ is strongly equivalent to $\top$.
    \\\noindent Case 2: $\langle r_1, \dots r_k \rangle \not\in [t_1,\dots,t_k]$.
    By definition, there exists an $r_i$ ($1 \leq i \leq k$) such that $r_i \not\in [t_i]$.
    By Proposition~\ref{prop:gr:val} it follows that the corresponding formula $\grdstd(val_{t_i}(r_i)) \equiv_s \bot$.
    Thus, $\grdstd(val_\boldt(\boldr))$ is strongly equivalent to $\bot$.
\end{proof}


%% file: appendix-3.tex
\subsection{Proof of Proposition~\ref{prop:cl}}


\begin{lemma}
    \label{lem:grd:val}
    Let $\intfolone$ be a standard interpretation and let $\boldp$ be the standard partition.
    Then, for any $n$-tuple of program variables $\boldZ$, $n$-tuple of ground program terms $\boldt$, some universe $|\intfolone|$ and formula $G$, 
    it follows that
    $$
    \{\grdstd(val_{\boldt}(\boldZ)_\boldr^\boldZ) \wedge \grdstd(G(\boldZ)_\boldr^\boldZ) \mid \boldr \in |\intfolone|^n \}^\vee
    \equiv_s 
    \{\grdstd(G(\boldr)) \mid \boldr \in |\intfolone|^n, \boldr \in [\boldt]\}^\vee
    $$
\end{lemma}

\begin{proof}
    First note that $\boldr$ is an $n$-tuple of precomputed terms, and that $\grdstd(val_{\boldt}(\boldZ)_\boldr^\boldZ)$ is more conveniently written as
    $\grdstd(val_{\boldt}(\boldr))$ (similarly, rewrite $\grdstd(G(\boldZ)_\boldr^\boldZ)$ as $\grdstd(G(\boldr))$).
    %
    %
    We now consider the unique partition of $|\intfolone|^n$ into sets $\bold{J}$ and $\bold{K}$ so that:
    $\bold{J}$ denotes the set of tuples of the form $\langle j_1,\dots,j_n \rangle$  such that  every $j_i \in [t_i]$ ($1 \leq i \leq n$).
    As a result,  every tuple $\langle k_1,\dots,k_n \rangle \in \bold{K}$ is such that for some $i$ ($1 \leq i \leq n$), $k_i \not\in [t_i]$.
    By construction of $\bold{J}$ and $\bold{K}$, 
    \begin{align*}
        &\{\grdstd(val_{\boldt}(\boldZ)_\boldr^\boldZ) \wedge \grdstd(G(\boldZ)_\boldr^\boldZ) \mid \boldr \in |\intfolone|^n \}^\vee\\
        &= \{\left( \grdstd(val_{\boldt}(\bold{j})) \wedge \grdstd(G(\bold{j})) \right), \left( \grdstd(val_{\boldt}(\bold{k})) \wedge \grdstd(G(\bold{k})) \right) \mid \bold{j} \in \bold{J}, \bold{k} \in \bold{K}\}^\vee
    \end{align*}
    From Corollary~\ref{prop:gr:val:corr}, it follows that $\grdstd(val_{\boldt}(\bold{j})) \equiv_s \top$ and $\grdstd(val_{\boldt}(\bold{k})) \equiv_s \bot$ for each $\bold{j} \in \bold{J}, \bold{k} \in \bold{K}$.
    Thus,
    \begin{align*}
        &\{\grdstd(val_{\boldt}(\boldZ)_\boldr^\boldZ) \wedge \grdstd(G(\boldZ)_\boldr^\boldZ) \mid \boldr \in |\intfolone|^n \}^\vee\\
        &\equiv_s \{\left( \top \wedge \grdstd(G(\bold{j})) \right), \left( \bot \wedge \grdstd(G(\bold{k})) \right) \mid \bold{j} \in \bold{J}, \bold{k} \in \bold{K}\}^\vee \hfill{ (Fact~\ref{fact:5:harrison})}\\
        &\equiv_s \{ \grdstd(G(\bold{j})) \mid \bold{j} \in \bold{J} \}^\vee \hfill{ \text{(Fact~\ref{fact:1}, construction of $\bold{J}$)}}\\
        &\equiv_s \{\grdstd(G(\boldr)) \mid \boldr \in |\intfolone|^n, \boldr \in [\boldt]\}^\vee
    \end{align*}
\end{proof}



\begin{lemma}
    \label{lem:gr:atomic:comparison.aux}
    Let $\intfolone$ be a standard interpretation and let $\boldp$ be the standard partition.
    If $t_1 \prec t_2$ is a closed comparison, then the formula
    \begin{gather}
        \label{eq:gr:atomic:comparison.aux}
        \{\grdstd(r_1 \prec r_2)
                \mid \langle r_1, r_2 \rangle \in |\intfolone|^{\sortsuper} \times |\intfolone|^{\sortsuper}, \langle r_1, r_2 \rangle \in [t_1,t_2] \}^\vee
    \end{gather}
    is strongly equivalent to $\tau(t_1 \prec t_2)$.
\end{lemma}

\begin{proof}
    Since $t_1 \prec t_2$ is a closed comparison and grounding is performed with respect to an empty set of intensional functions,
    it follows from conditions 3 and 4 of grounding that~\eqref{eq:gr:atomic:comparison.aux}
    is strongly equivalent to
    \begin{align*}
        \{\top \mid r_1 \in [t_1], r_2 \in [t_2], \prec(r_1,r_2) \}^\vee \vee
        \{\bot \mid r_1 \in [t_1], r_2 \in [t_2], \not\prec(r_1,r_2) \}^\vee
    \end{align*}
    which, by Fact~\ref{fact:1}, is strongly equivalent to
    $\{\top \mid r_1 \in [t_1], r_2 \in [t_2], \prec(r_1,r_2) \}^\vee.$
    Therefore, it follows that~\eqref{eq:gr:atomic:comparison.aux} is strongly equivalent to $\top$
    if some $\langle r_1, r_2 \rangle \in [t_1,t_2]$ satisfies $r_1 \prec r_2$ and to $\bot$ otherwise.
    Therefore, the result follows from condition 5 of the definition of~$\tau$. 
\end{proof}

\begin{lemma}
    \label{lem:lit}
    Let $l$ be a basic literal or comparison, 
    let $\boldZ$ be a superset of the global variables of $l$, 
    and let $\boldz$ be a tuple of precomputed terms of the same length as $\boldz$.
    Then, for a standard interpretation $\intfolone$ and the standard partition $\boldp$,
    $$\grdstd \left((\tau^B_{\boldZ}(l))_\boldz^\boldZ \right) \equiv_s \tau \left( l_\boldz^\boldZ \right).$$
\end{lemma}

\begin{proof}
    We proceed by cases.
    \\[5pt]
    \\\noindent \emph{Case 1:} $l$ is an atom, $p(t_1,\dots,t_k)$ where $p$ is an intensional symbol ($p \in \boldp$).
    Then, 
    $$
    \tau^B_{\boldZ}(l) = \exists V_1, \dots V_k (val_{t_1}(V_1) \wedge \dots \wedge val_{t_k}(V_k) \wedge p(V_1,\dots,V_k))
    $$
    which we will abbreviate as 
    $$
    \tau^B_{\boldZ}(l) = \exists \boldV (val_{\boldt}(\boldV) \wedge p(\boldV))
    $$
    thus,
    \begin{align*}
        \grdstd(\tau^B_{\boldZ}\left( l \right)_{\boldz}^{\boldZ})
        &=\grdstd(
            \exists \boldV ( val_{\boldt}(\boldV) \wedge p(\boldV) )_{\boldz}^{\boldZ} ) \\
        &=\grdstd(
            \exists \boldV ( val_{\boldt_{\boldz}^{\boldZ}}(\boldV) \wedge p(\boldV) ) ) \\
        &=\{ \grdstd(
                val_{\boldt_{\boldz}^{\boldZ}}(\boldV)_{\boldr}^{\boldV} \wedge p(\boldr)
            ) \mid \boldr \in |\intfolone|^{\sortsuper} \}^\vee  \\
        &=\{ \grdstd(val_{\boldt_{\boldz}^{\boldZ}}(\boldV)_{\boldr}^{\boldV})
                \wedge \grdstd(p(\boldr))
                    \mid \boldr \in |\intfolone|^{\sortsuper} \}^\vee  \\
        &\equiv_s \{ \grdstd(p(\boldr))
                    \mid \boldr \in |\intfolone|^{\sortsuper}, \boldr \in [\boldt_{\boldz}^{\boldZ}] \}^\vee \hfill \text{     (Lemma~\ref{lem:grd:val})}  \\
        &= \{ p(\boldr) \mid \boldr \in |\intfolone|^{\sortsuper}, \boldr \in [\boldt_{\boldz}^{\boldZ}] \}^\vee \hfill{\text{  (Definition of } \grdstd \text{)}} \\
        &= \tau \left( p(\boldt_{\boldz}^{\boldZ}) \right)  ~\hfill{\text{(Definition of $\tau$, condition 2)}}\\
        &= \tau(l_{\boldz}^{\boldZ}) 
    \end{align*}
    \\[5pt]
    \noindent\emph{Cases 2,3:} $l$ is a singly or doubly negated intensional atom. The proof of these cases mimics the proof of Case 1.
    \\[5pt]
    \\\noindent \emph{Case 4:} $l$ is a comparison, $a \prec b$.
    Then, 
    $$
    \tau^B_{\boldZ}(l) = \exists V_1, V_2 (val_{a}(V_1) \wedge val_{b}(V_2) \wedge V_1 \prec V_2)
    $$
    thus,
    \begin{align*}
        &~~~~~\grdstd(\tau^B_{\boldZ}\left( l \right)_{\boldz}^{\boldZ})\\
        &=\grdstd( \exists V_1, V_2 (val_{a}(V_1) \wedge val_{b}(V_2) \wedge V_1 \prec V_2)_{\boldz}^{\boldZ}) \\
        &=\grdstd(
            \exists V_1 V_2 ( val_{a_{\boldz}^{\boldZ}}(V_1) \wedge val_{b_{\boldz}^{\boldZ}}(V_2) \wedge V_1 \prec V_2 ) ) \\
        &=\{\grdstd(
            val_{a_{\boldz}^{\boldZ}}(V_1)_{r_1}^{V_1} \wedge val_{b_{\boldz}^{\boldZ}}(V_2)_{r_2}^{V_2} \wedge (V_1 \prec V_2)_{r_1,r_2}^{V_1,V_2})
                \mid \langle r_1, r_2 \rangle \in |\intfolone|^{\sortsuper} \times |\intfolone|^{\sortsuper} \}^\vee \\
        &\equiv_s\{\grdstd(r_1 \prec r_2)
                \mid \langle r_1, r_2 \rangle \in |\intfolone|^{\sortsuper} \times |\intfolone|^{\sortsuper}, \langle r_1, r_2 \rangle \in [a_{\boldz}^{\boldZ},b_{\boldz}^{\boldZ}] \}^\vee \\
        &\equiv_s \tau(a_{\boldz}^{\boldZ} \prec b_{\boldz}^{\boldZ}) \text{ (Lemma~\ref{lem:gr:atomic:comparison.aux})} \\
        &= \tau(l_{\boldz}^{\boldZ}) 
    \end{align*}
\end{proof}

\paragraph{Proof of Proposition~\ref{prop:cl}}
Since $\boldZ$ is a superset of the global variables of our conditional literal,
it follows that $(\clhead : \boldL)_{\boldz}^{\boldZ}$ is a closed expression (only non-global variables remain).
%
%
%

By definition, $\tau^B_\boldZ(\clhead : \boldL)$ is
\begin{gather*}
    \forall \boldX \left(\tau^B_\boldZ(\boldL) \rar \tau^B_\boldZ(\clhead)\right)
\end{gather*}
where $\boldX$ denotes the variables of $\clhead : \boldL$ that do not occur in $\boldZ$, 
observe that $\boldX$ is also the set of all variables occurring in $(\clhead : \boldL)_{\boldz}^{\boldZ}$.
Thus, for any expression $E$ occurring in $\clhead : \boldL$ and tuple of precomputed terms $\boldx$ of the same length as $\boldX$,
$E^{\boldX\boldZ}_{\boldx\boldz}$ is a ground expression.
Similarly, for any expression $E$ occurring in $(\clhead : \boldL)^{\boldZ}_{\boldz}$,
$E^{\boldX}_{\boldx}$ is a ground expression.
Thus,
\begin{align*}
    &\grdstd((\tau^B_{\boldZ}\left(\clhead : \boldL\right))_{\boldz}^{\boldZ}) \\
    &= \grdstd(\forall \boldX \left(\tau^B_\boldZ(\boldL) \rar \tau^B_\boldZ(\clhead)\right)_{\boldz}^{\boldZ}) \\
    &= \grdstd(\forall \boldX \left(\tau^B_\boldZ(l_1) \wedge \dots \wedge \tau^B_\boldZ(l_m) \rar \tau^B_\boldZ(\clhead)\right)_{\boldz}^{\boldZ})\\
    &= \{\grdstd{\left(\tau^B_{\boldZ}(l_1)_{\boldz}^{\boldZ} \wedge
        \dots \wedge \tau^B_{\boldZ}(l_m)_{\boldz}^{\boldZ}
        \rar \tau^B_{\boldZ}(H)_{\boldz}^{\boldZ}\right)_{\boldx}^{\boldX}} \mid \boldx \in |\intfolone|^{\sortsuper} \}^\wedge \\
    &= \{\grdstd{(\tau^B_{\boldZ}(l_1)_{\boldz}^{\boldZ})_{\boldx}^{\boldX}} \wedge
        \dots \wedge \grdstd{(\tau^B_{\boldZ}(l_m)_{\boldz}^{\boldZ})_{\boldx}^{\boldX}}
        \rar \grdstd{(\tau^B_{\boldZ}(H)_{\boldz}^{\boldZ})_{\boldx}^{\boldX}} \mid \boldx \in |\intfolone|^{\sortsuper} \}^\wedge
\end{align*}
In the preceding, each expression
$(\tau^B_{\boldZ}(l_i)_{\boldz}^{\boldZ})_{\boldx}^{\boldX}$
where $l_i \in \{l_1, \dots, l_m\}$
denotes the expression obtained from the first-order formula $\tau^B_{\boldZ}(l_i)$
by first substituting precomputed terms $\boldz$ for variables in $\boldZ$,
and then substituting precomputed terms $\boldx$ for any of the remaining variables that occur in $\boldX$.

But, since $\boldX$ and $\boldZ$ are disjoint, the order of replacement doesn't matter.
That is,
$
(E_{\boldz}^{\boldZ})_{\boldx}^{\boldX}
= E_{\boldz\boldx}^{\boldZ\boldX}
= (E_{\boldx}^{\boldX})_{\boldz}^{\boldZ}
$.
Since $\boldX \cup \boldZ$ is a superset of the (global) variables of $l_i$,
the resulting expression is a ground basic literal.
The same is true for $(\tau^B_{\boldZ}(H)_{\boldz}^{\boldZ})_{\boldx}^{\boldX}$
when $H$ is a basic literal.

Now observe that, since $(\clhead : \boldL)_{\boldz}^{\boldZ}$ is closed, $\tau((\clhead : \boldL)_{\boldz}^{\boldZ})$ is
\begin{align*}
    & \{ \tau(\boldL_{\boldz\boldx}^{\boldZ\boldX}) \rar \tau(\clhead_{\boldz\boldx}^{\boldZ\boldX}) \mid \boldx \in |\intfolone|^{\sortsuper} \}^\wedge \\
    =& \{ \tau([l_1]_{\boldz\boldx}^{\boldZ\boldX}) \wedge \dots \wedge \tau([l_m]_{\boldz\boldx}^{\boldZ\boldX}) \rar \tau(\clhead_{\boldz\boldx}^{\boldZ\boldX}) \mid \boldx \in |\intfolone|^{\sortsuper} \}^\wedge.
\end{align*}
It follows from Lemma~\ref{lem:lit}
that for any basic literal $l_i \in \{l_1, \dots, l_m\}$,
$$
    \grdstd((\tau^B_{\boldZ}(l_i)_{\boldz}^{\boldZ})_{\boldx}^{\boldX})
    = \grdstd(\tau^B_{\boldZ}(l_i)_{\boldz\boldx}^{\boldZ\boldX})
    \equiv_s \tau([l_i]_{\boldz\boldx}^{\boldZ\boldX}).
$$
The same holds for $H$ when $H$ is a basic literal.
When $H$ is the symbol $\bot$, observe that
$$
    \grdstd((\tau^B_{\boldZ}(H)_{\boldz}^{\boldZ})_{\boldx}^{\boldX})
    = \grdstd(\bot)
    = \bot
    = \tau(\bot)
    = \tau(H_{\boldz\boldx}^{\boldZ\boldX}).
$$
From these conclusions and from Fact~\ref{fact:5:harrison}, it follows that
\begin{align*}
    & ~\{\grdstd{(\tau^B_{\boldZ}(l_1)_{\boldz}^{\boldZ})_{\boldx}^{\boldX}} \wedge
        \dots \wedge \grdstd{(\tau^B_{\boldZ}(l_m)_{\boldz}^{\boldZ})_{\boldx}^{\boldX}}
        \rar \grdstd{(\tau^B_{\boldZ}(H)_{\boldz}^{\boldZ})_{\boldx}^{\boldX}} \mid \boldx \in |\intfolone|^{\sortsuper} \}^\wedge\\
    \equiv_s& ~\{ \tau([l_1]_{\boldz\boldx}^{\boldZ\boldX}) \wedge \dots \wedge \tau([l_m]_{\boldz\boldx}^{\boldZ\boldX}) \rar \tau(\clhead_{\boldz\boldx}^{\boldZ\boldX}) \mid \boldy \in |\intfolone|^{\sortsuper} \}^\wedge.
\end{align*}
from which the result follows.



%% file: appendix-4.tex
\subsection{Proofs of Propositions~\ref{prop:rule} and~\ref{prop:program}}
In the following, it is convenient to abuse notation for a $k$-tuple $\boldv$ composed of terms  from a universe $|\intfolone|$: we write $\boldv \in |\intfolone|$ to denote  $\boldv \in |\intfolone|^k$.

The proof of Proposition~\ref{prop:rule} relies on a rule \emph{instantiation} process used to obtain closed rules.
Let~$\boldZ$ denote the global variables of rule $R$.
By $inst_\oc(R)$ we denote the set of  instances of rule $R$ w.r.t. a set $\oc$ of precomputed terms, i.e.,
$$inst_\oc(R) = \{R_\boldz^\boldZ \mid \boldz \in \oc\}=\{ [H]_\boldz^\boldZ \ruleo [B_1]_\boldz^\boldZ, \dots, [B_n]_\boldz^\boldZ \mid \boldz \in \oc\}.$$
If we take $\oc$ to be the universe of precomputed terms, then 
 $\tau(R) = \{\tau(\rho) \mid \rho \in inst_{\oc}(R)\}^{\wedge}$.

\begin{lemma}
    \label{lem:rule.head.tau.helper}
    Let $\boldv$ be a tuple of precomputed terms,  let $F$ and $G$ 
    be infinitary propositional formulas, and
     let $\boldt$ be a tuple of ground terms of the same length as $\boldv$.
    For a tuple of precomputed terms $\boldr$ of the same length as $\boldv$,   
    \begin{gather}
        \label{eq:lem:rule.head.tau.helper.left}
        \{ F \to G^\boldv_\boldr \mid \boldr \in [\boldt]\}^\wedge
    \end{gather}
    is strongly equivalent to
    \begin{gather}
        \label{eq:lem:rule.head.tau.helper.right}
        F \to \bigwedge_{\boldr \in [\boldt]} G^\boldv_\boldr
    \end{gather}
\end{lemma}

\begin{proof}
    Theorem~3 by~\cite{harlifpeval17} states that two (sets of) infinitary formulas are strongly equivalent if and only if they are equivalent in the infinitary logic of here-and-there. 
    This proof will illustrate that infinitary formulas~\eqref{eq:lem:rule.head.tau.helper.left} and~\eqref{eq:lem:rule.head.tau.helper.right} are equivalent in the infinitary logic of here-and-there.

    Take any \htinterp\ $\hti$. We will consider two cases. In the first case, we assume that $\hti \models \eqref{eq:lem:rule.head.tau.helper.left}$ and illustrate that from this assumption it follows that  $\hti \models \eqref{eq:lem:rule.head.tau.helper.right}$.  
    In the second case, we assume that $\hti \models\eqref{eq:lem:rule.head.tau.helper.right}$  and illustrate that from this assumption it follows that   $\hti \models \eqref{eq:lem:rule.head.tau.helper.left}$. 
    These arguments illustrate that formulas~\eqref{eq:lem:rule.head.tau.helper.left} and~\eqref{eq:lem:rule.head.tau.helper.right} are equivalent in the infinitary logic of here-and-there.
    \\\noindent\emph{Case 1:} 
    Assume $\hti \models \eqref{eq:lem:rule.head.tau.helper.left}$.
    It follows that $\hti \models F \to G^{\boldv}_{\boldr}$ for any $\boldr \in [\boldt]$.
    From the definition of ht-satisfaction, we conclude that for any $\boldr \in [\boldt]$ 
    \begin{itemize}
    \item $\intproptwo \models F \to G^{\boldv}_{\boldr}$ and 
    \item $\hti \not\models F$ or $\hti \models G^{\boldv}_{\boldr}$.    
    \end{itemize}
    We now illustrate that $\intproptwo \models F \to \bigwedge_{\boldr \in [\boldt]} G^{\boldv}_{\boldr}$.
    For the case when  $\intproptwo \not\models F$, the statement vacuously holds.
    Now assume that $\intproptwo \models F$.
    Since $\intproptwo \models F \to G^{\boldv}_{\boldr}$ for any $\boldr \in [\boldt]$, it follows that $\intproptwo \models G^{\boldv}_{\boldr}$ for any $\boldr \in [\boldt]$.
    Equivalently, $\intproptwo \models \bigwedge_{\boldr \in [\boldt]} G^{\boldv}_{\boldr}$.
    Thus, $\intproptwo \models F \to \bigwedge_{\boldr \in [\boldt]} G^{\boldv}_{\boldr}$.
    \\[5pt]
    Since $\hti \not\models F$ or $\hti \models G^{\boldv}_{\boldr}$ for any $\boldr \in [\boldt]$, and since the satisfaction of $F$ is independent of $\boldr$,
    it follows that 
    \begin{itemize}
    \item $\hti \not\models F$ or 
    \item $\hti \models G^{\boldv}_{\boldr}$ for any $\boldr \in [\boldt]$.
    \end{itemize}
    Consequently, $\hti \not\models F$ or $\hti \models \bigwedge_{\boldr \in [\boldt]} G^{\boldv}_{\boldr}$.
    \\[5pt]
    It now follows from the definition of ht-satisfaction that $\hti \models \eqref{eq:lem:rule.head.tau.helper.right}$.

    \noindent \emph{Case 2:}
    Assume $\hti \models \eqref{eq:lem:rule.head.tau.helper.right}$.
    It follows that 
    \begin{itemize}
        \item $\intproptwo \models F \to \bigwedge_{\boldr \in [\boldt]} G^{\boldv}_{\boldr}$ and
        \item $\hti \not\models F$ or $\hti \models \bigwedge_{\boldr \in [\boldt]} G^{\boldv}_{\boldr}$.
    \end{itemize}
    To illustrate that $\hti \models \eqref{eq:lem:rule.head.tau.helper.left}$, it is sufficient to show that $\hti \models F \to G^\boldv_{\bold{n}}$ for an arbitrary $\bold{n} \in [\boldt]$.
    \\[5pt]
    \\Now assume that $\intproptwo \models F$.
    It follows that
    $\intproptwo \models \bigwedge_{\boldr \in [\boldt]} G^{\boldv}_{\boldr}$, therefore, $\intproptwo \models G^{\boldv}_{\bold{n}}$.
    Thus, $\intproptwo \models F \to G^{\boldv}_{\bold{n}}$.
    \\[5pt]
    Since $\hti \not\models F$ or $\hti \models \bigwedge_{\boldr \in [\boldt]} G^{\boldv}_{\boldr}$, it follows that
    $\hti \not\models F$ or $\hti \models G^{\boldv}_{\bold{n}} \wedge \bigwedge_{\boldr \in [\boldt] \setminus \bold{n}} G^{\boldv}_{\boldr}$.
    Therefore, $\hti \not\models F$ or $\hti \models G^{\boldv}_{\bold{n}}$.
    \\[5pt]
    Thus, $\hti \models F \to G^{\boldv}_{\bold{n}}$ for arbitrary $\bold{n} \in [\boldt]$.
    Consequently, $\hti \models \eqref{eq:lem:rule.head.tau.helper.left}$.
\end{proof}

The claim  below follows from the definition of ht- and classical satisfaction using well known classical equivalences.
\begin{lemma}
    \label{lem:rule.head.choice.tau.helper}
    Let $F$ and $G$ be infinitary propositional formulas.
    Then, $$F \wedge \neg\neg G \to G \equiv_s F \to G \vee \neg G.$$
\end{lemma}


\begin{lemma}
    \label{lem:term.replacement}
    Let $t$ be a term from a rule containing global variables $\boldZ$, and let $\boldz$ be a tuple of precomputed terms of the same length as $\boldZ$.
    Let $V$ be a program variable which does not occur in $\boldZ$ or $t$, and let $v$ be a precomputed term.
    Then, 
    \begin{gather*}
        \left( val_t(V)_v^V \right)_\boldz^\boldZ = \left( val_{t_\boldz^\boldZ}(V) \right)_v^V
    \end{gather*}
\end{lemma}

\begin{proof}
    The proof uses structural induction across the forms of terms; we say that a term $t$ has the induction hypothesis property when the following equivalence holds for $t$ and a program variable $V$ which does not occur in $\boldZ$:
    $$
    val_t(V)_\boldz^\boldZ = val_{t_\boldz^\boldZ}(V)
    $$
    We illustrate the first two cases here.
    \\[5pt]
    \emph{Case 1:} $t$ is a numeral, symbolic constant, variable, $\mathit{inf}$ or~$\mathit{sup}$.
    \begin{align*}
        \left( val_t(V)_v^V \right)_\boldz^\boldZ 
        &= \left( (V = t)_v^V \right)_\boldz^\boldZ \\
        &= \left( v = t \right)_\boldz^\boldZ \\
        &= \left( v = t_\boldz^\boldZ \right) \\
        &= \left( V = t_\boldz^\boldZ \right)_v^V \\
        &= \left( val_{t_\boldz^\boldZ}(V) \right)_v^V
    \end{align*}
    \\[5pt]
    \emph{Case 2:} $t$ is $|t_1|$, where $|t_1|$ has the induction hypothesis property.
    \begin{align*}
        \left( val_t(V)_v^V \right)_\boldz^\boldZ 
        &= \left( \exists I \left(val_{t_1}(I) \wedge V = |I| \right)_v^V \right)_\boldz^\boldZ \\
        &= \left( \exists I \left(val_{t_1}(I) \wedge v = |I| \right) \right)_\boldz^\boldZ   \\
        &= \exists I \left( val_{t_1}(I)_\boldz^\boldZ \wedge (v = |I|)_\boldz^\boldZ \right) \\
        &= \exists I \left( val_{{t_1}_\boldz^\boldZ}(I) \wedge v = |I| \right)~~~~~~~~\text{(induction hypothesis)} \\
        &= \left( \exists I \left( val_{{t_1}_\boldz^\boldZ}(I) \wedge V = |I| \right) \right)_v^V \\
        &= \left( val_{t_\boldz^\boldZ}(V) \right)_v^V
    \end{align*}
\end{proof}

\begin{lemma}
    \label{lem:remove.grounding}
    Let $\boldt$ be a $k$-tuple of ground terms, $\boldV$ be a $k$-tuple of program variables,  $F, G$ be infinitary propositional formulas, $\intfolone$ be a standard interpretation,  and $\boldp$ be the standard partition. It holds that
    \begin{align*}
        &\{ \grdstd(val_\boldt(\boldV))_\boldv^\boldV \wedge F  \rar G  \mid \boldv \in |\intfolone|^{\sortsuper}\}^\wedge \equiv_s 
        \\
        &\{ F  \rar G  \mid \boldv \in |\intfolone|^{\sortsuper}, \boldv \in [\boldt]\}^\wedge 
    \end{align*}
\end{lemma}

\begin{proof}
    \begin{align*}
                 &\{ \grdstd(val_\boldt(\boldV))_\boldv^\boldV \wedge F \rar G  \mid \boldv \in |\intfolone|^{\sortsuper}\}^\wedge \\
        = &\{ \grdstd(val_\boldt(\boldv)) \wedge F  \rar G  \mid \boldv \in |\intfolone|^{\sortsuper}\}^\wedge \\
        \equiv_s &\{ \grdstd(val_\boldt(\boldv)) \wedge F \rar G  \mid \boldv \in |\intfolone|^{\sortsuper}, \boldv \in [\boldt]\}^\wedge
        \wedge \\
        &\{ \grdstd(val_\boldt(\bold{w})) \wedge F \rar G  \mid \bold{w} \in |\intfolone|^{\sortsuper}, \bold{w} \not\in [\boldt]\}^\wedge \\
        \equiv_s &\{ \top \wedge F \rar G  \mid \boldv \in |\intfolone|^{\sortsuper}, \boldv \in [\boldt]\}^\wedge
        \wedge \\
        &\{ \bot \wedge F \rar G  \mid \bold{w} \in |\intfolone|^{\sortsuper}, \bold{w} \not\in [\boldt]\}^\wedge ~~~~~~~~~~~~~~~~~~~~~~~~~~~~~~~~~~~~~~~~~~~~~\hfill{\text{Corollary~\ref{prop:gr:val:corr}}}\\
        \equiv_s &\{ F \rar G  \mid \boldv \in |\intfolone|^{\sortsuper}, \boldv \in [\boldt]\}^\wedge
        \wedge \{ \bot \rar G  \mid \bold{w} \in |\intfolone|^{\sortsuper}, \bold{w} \not\in [\boldt]\}^\wedge ~\hfill{\text{Fact~\ref{fact:1}: condition~\ref{l:1:2}}}\\
        \equiv_s &\{ F \rar G  \mid \boldv \in |\intfolone|^{\sortsuper}, \boldv \in [\boldt]\}^\wedge
        \wedge \{ \top  \mid \bold{w} \in |\intfolone|^{\sortsuper}, \bold{w} \not\in [\boldt]\}^\wedge ~~~~~~~~~~\hfill{\text{Fact~\ref{fact:1}: condition~\ref{l:1:7}}}\\
        \equiv_s &\{ F \rar G  \mid \boldv \in |\intfolone|^{\sortsuper}, \boldv \in [\boldt]\}^\wedge~~~~~~~~~~~~~~~~~~~~~~~~~~~~~~~~~~~~~~~~~~~~~~~~~~~~~~~\hfill{\text{Fact~\ref{fact:1}: condition~\ref{l:1:1}}}
    \end{align*}
\end{proof}

\paragraph{Proof of Proposition~\ref{prop:rule}}

\begin{proof}
    We proceed by cases.
    \\[5pt]
    \\\noindent \emph{Case 1}: $R$ is a normal rule with atom $p(\boldt)$ in the head.
    By definition,
    \begin{align*}
        \grdstd\left(\tau^*(R)\right) 
        = \grdstd\left(\forall \boldV \left( \forall \boldZ \left(val_\boldt(\boldV) \wedge \tau^*_{\boldZ}(B_1) \wedge \dots \wedge \tau^*_{\boldZ}(B_n) \to p(\boldV)\right)\right)\right)
    \end{align*}
    where $\boldV$ is a tuple of fresh, alphabetically first program variables disjoint from the global variables $\boldZ$ of the same length as $\boldt$.
    Let $\boldz$ denote a tuple of precomputed terms of the same length as $\boldZ$, and
    let $\boldv$ denote a tuple of precomputed terms of the same length as $\boldV$.
    Then, $\grdstd\left(\tau^*(R)\right)$ is equal to
    \begin{align*}
        &\big\{ \grdstd\left( \forall \boldZ ( val_\boldt(\boldv) \wedge \tau^*_{\boldZ}(B_1) \wedge \dots \wedge \tau^*_{\boldZ}(B_n) \to p(\boldV)_\boldv^\boldV ) \right) \mid \boldv \in |\intfolone|^{\sortsuper} \big\}^\wedge\\
        &= \big\{ \{ \grdstd \left( val_\boldt(\boldv)_\boldz^\boldZ \wedge \tau^*_{\boldZ}(B_1)_\boldz^\boldZ  \wedge \dots \wedge \tau^*_{\boldZ}(B_n)_\boldz^\boldZ  \to p(\boldv)_{\boldz}^{\boldZ}  \right)  \mid \boldz \in |\intfolone|^{\sortsuper}\}^\wedge\mid \boldv \in |\intfolone|^{\sortsuper} \big\}^\wedge\\
        &= \big\{ \{ \grdstd \left(val_\boldt(\boldv)_\boldz^\boldZ\right) \wedge \grdstd \left(\tau^*_{\boldZ}(B_1)_\boldz^\boldZ\right)  \wedge \dots \wedge \grdstd \left(\tau^*_{\boldZ}(B_n)_\boldz^\boldZ \right)  \to \grdstd \left(p(\boldv) \right)\\
        &~~~~~~~~~~~~~~~~~~~~~~~~~~~~~~~~~~~~~~~~~~~~~~~~~~~~~~~~~~~~~~~~~~~~~~~~~~~~~~~~~~~~~~~~~~~\mid \boldz \in |\intfolone|^{\sortsuper}\}^\wedge\mid \boldv \in |\intfolone|^{\sortsuper} \big\}^\wedge\\
        &= \big\{ \{ \grdstd \left(val_\boldt(\boldv)_\boldz^\boldZ \right) \wedge \grdstd \left(\tau^*_{\boldZ}([B_1]_\boldz^\boldZ) \right)  \wedge \dots \wedge \grdstd \left(\tau^*_{\boldZ}([B_n]_\boldz^\boldZ) \right) \to p(\boldv)\\ 
        &~~~~~~~~~~~~~~~~~~~~~~~~~~~~~~~~~~~~~~~~~~~~~~~~~~~~~~~~~~~~~~~~~~~~~~~~~~~~~~~~~~~~~~~~~~~\mid \boldz \in |\intfolone|^{\sortsuper}\}^\wedge\mid \boldv \in |\intfolone|^{\sortsuper} \big\}^\wedge 
    \end{align*}
    From the preceding and the definition of ht-satisfaction, it follows that $\grdstd\left(\tau^*(R)\right)$ is strongly equivalent to
    $$
        \{ \grdstd(val_\boldt(\boldv)_\boldz^\boldZ) \wedge \grdstd(\tau^*_{\boldZ}([B_1]_\boldz^\boldZ))  \wedge \dots \wedge \grdstd(\tau^*_{\boldZ}([B_n]_\boldz^\boldZ))  \to p(\boldv)  \mid \boldz \in |\intfolone|^{\sortsuper}, \boldv \in |\intfolone|^{\sortsuper}\}^\wedge 
    $$
    Note that each expression $[B_i]^{\boldZ}_{\boldz}$ is a closed conditional literal.
    Thus, from Proposition~\ref{prop:cl} it follows that
    $$\grdstd(\tau^*_{\boldz}([B_i]^{\boldZ}_{\boldz})) \equiv_s \tau([B_i]^{\boldZ}_{\boldz})$$
    and so, from Fact~\ref{fact:5:harrison}, $\grdstd\left(\tau^*(R)\right)$ is strongly equivalent to
    \begin{align*}
         \{ \grdstd(val_\boldt(\boldv)_\boldz^\boldZ) \wedge \tau([B_1]_\boldz^\boldZ)  \wedge \dots \wedge \tau([B_n]_\boldz^\boldZ)  \to p(\boldv)  \mid \boldz \in |\intfolone|^{\sortsuper}, \boldv \in |\intfolone|^{\sortsuper}\}^\wedge. 
    \end{align*}
    From Lemma~\ref{lem:term.replacement}, which is generalized in a straightforward way to tuples of terms and variables, we conclude that $\grdstd(val_\boldt(\boldv)_\boldz^\boldZ) = \grdstd(val_{\boldt_\boldz^\boldZ}(\boldv))$ for any $\boldv \in |\intfolone|^{\sortsuper}$.
    Thus, it follows that  $\grdstd\left(\tau^*(R)\right)$ is strongly equivalent to
    $$    \{ \grdstd(val_{\boldt_\boldz^\boldZ}(\boldv)) \wedge \tau([B_1]_\boldz^\boldZ)  \wedge \dots \wedge \tau([B_n]_\boldz^\boldZ)  \to p(\boldv)  \mid \boldz \in |\intfolone|^{\sortsuper}, \boldv \in |\intfolone|^{\sortsuper}\}^\wedge. $$
    %
    From Lemma~\ref{lem:remove.grounding}, it follows that $\grdstd\left(\tau^*(R)\right)$ is strongly equivalent to
    \begin{gather}
        \label{eq:gr.tau.rule}
        \{ \tau([B_1]_\boldz^\boldZ)  \wedge \dots \wedge \tau([B_n]_\boldz^\boldZ)  \to p(\boldv)  \mid \boldz \in |\intfolone|^{\sortsuper}, \boldv \in |\intfolone|^{\sortsuper}, \boldv \in [\boldt_\boldz^\boldZ]\}^\wedge.
    \end{gather}
    Since $R$ is a normal rule with the atom $p(\boldt)$ in the head, the rule instantiation process for $R$ produces the following:
    \begin{align*}
        \tau(R) 
        &= \{ \tau \left( p(\boldt)_\boldz^\boldZ \ruleo [B_1]_\boldz^\boldZ,\dots,[B_n]_\boldz^\boldZ \right) \mid  \boldz \in |\intfolone|^{\sortsuper}\}^\wedge\\
        &= \{ \tau \left( p(\boldt_\boldz^\boldZ) \ruleo [B_1]_\boldz^\boldZ,\dots,[B_n]_\boldz^\boldZ \right) \mid  \boldz \in |\intfolone|^{\sortsuper}\}^\wedge\\
        &= \{\tau([B_1]_\boldz^\boldZ) \wedge \dots \wedge \tau([B_n]_\boldz^\boldZ) \to \bigwedge_{\boldr \in [\boldt_\boldz^\boldZ]}p(\boldr) \mid \boldz \in |\intfolone|^{\sortsuper}\}^\wedge
    \end{align*}
    It only remains to be shown that~\eqref{eq:gr.tau.rule} is strongly equivalent to the preceding formula.
    In the sequel we abbreviate expression  $\tau([B_1]_\boldz^\boldZ) \wedge \dots \wedge \tau([B_n]_\boldz^\boldZ)$
    by
    $\bold{B}_\boldz^\boldZ$.
    %
    %
    
    Now take tuple of precomputed terms $\boldz \in |\intfolone|^{\sortsuper}$ of length $\boldZ$.
    It remains  to show that 
    $$
    \{ \bold{B}_\boldz^\boldZ \to p(\boldv)  \mid \boldv \in |\intfolone|^{\sortsuper}, \boldv \in [\boldt_\boldz^\boldZ]\}^\wedge \equiv_s 
       \bold{B}_\boldz^\boldZ \to \bigwedge_{\boldr \in [\boldt_\boldz^\boldZ]}p(\boldr).
    $$
    Since $\bold{B}_\boldz^\boldZ$ is an infinitary formula  and since $\boldt_\boldz^\boldZ$ is a tuple of ground terms of the same length as $\boldv$, this result follows directly from Lemma~\ref{lem:rule.head.tau.helper}. 
    %
    %
    %
    %
    %
    \\[5pt]
    \\\noindent \emph{Case 2}: $R$ is a choice rule with the head of the form $\{p(\boldt)\}$.
    We take 
    $\boldV$, $\boldZ$, $\boldt$, $\boldz$,  $\boldv$ to denote the same entities as in Case 1.
Now,
    \begin{align*}
        &\grdstd\left(\tau^*(R)\right)\\ 
        &= \grdstd\Big(\forall \boldV \boldZ \big(val_\boldt(\boldV) \wedge \tau^*_{\boldZ}(B_1) \wedge \dots \wedge \tau^*_{\boldZ}(B_n) \wedge \neg\neg p(\boldV) \to p(\boldV)\big)\Big) \\
        &= \{ \grdstd\left( \forall \boldZ ( val_\boldt(\boldv) \wedge \tau^*_{\boldZ}(B_1) \wedge \dots \wedge \tau^*_{\boldZ}(B_n) \wedge \neg\neg p(\boldv)\to p(\boldv) ) \right) \mid \boldv \in |\intfolone|^{\sortsuper}\}^\wedge \\
        &= \{ \grdstd(val_\boldt(\boldv)_\boldz^\boldZ) \wedge \grdstd(\tau^*_{\boldZ}(B_1)_\boldz^\boldZ)  \wedge \dots \wedge \grdstd(\tau^*_{\boldZ}(B_n)_\boldz^\boldZ) \wedge \grdstd(\neg \neg p(\boldv))\\
        &\quad\quad\quad\quad\quad\quad\quad\quad\quad\quad\quad\quad~~~~~~~~~~~~~~\to \grdstd(p(\boldv)) \mid \boldz \in |\intfolone|^{\sortsuper}\}^\wedge, \boldv \in |\intfolone|^{\sortsuper}\}^\wedge \\
        &= \{ \grdstd(val_\boldt(\boldv)_\boldz^\boldZ) \wedge \grdstd(\tau^*_{\boldZ}([B_1]_\boldz^\boldZ))  \wedge \dots \wedge \grdstd(\tau^*_{\boldZ}([B_n]_\boldz^\boldZ)) \wedge \neg \neg \grdstd(p(\boldv))\\
        &\quad\quad\quad\quad\quad\quad\quad\quad\quad\quad\quad\quad~~~~~~~~~~~~~~~\to p(\boldv) \mid \boldz \in |\intfolone|^{\sortsuper}, \boldv \in |\intfolone|^{\sortsuper}\}^\wedge \\
        &= \{ \grdstd(val_\boldt(\boldv)_\boldz^\boldZ) \wedge \grdstd(\tau^*_{\boldZ}([B_1]_\boldz^\boldZ))  \wedge \dots \wedge \grdstd(\tau^*_{\boldZ}([B_n]_\boldz^\boldZ)) \wedge \neg \neg p(\boldv)\\
        &\quad\quad\quad\quad\quad\quad\quad\quad\quad\quad\quad\quad~~~~~~~~~~~~~~~~\to p(\boldv) \mid \boldz \in |\intfolone|^{\sortsuper}, \boldv \in |\intfolone|^{\sortsuper}\}^\wedge
    \end{align*}
    From Proposition~\ref{prop:cl} and Fact~\ref{fact:5:harrison}, it follows that $\grdstd\left(\tau^*(R)\right)$ is strongly equivalent to
    \begin{align*}
         &\{ \grdstd(val_\boldt(\boldv)_\boldz^\boldZ) \wedge  \tau \left( [B_1]_\boldz^\boldZ \right) \wedge \dots \wedge \tau\left( [B_n]_\boldz^\boldZ \right) \wedge \neg \neg p(\boldv)  \to p(\boldv) \mid \boldz \in |\intfolone|^{\sortsuper}, \boldv \in |\intfolone|^{\sortsuper}\}^\wedge\\
        &\equiv_s \{ \tau \left( [B_1]_\boldz^\boldZ \right) \wedge \dots \wedge \tau\left( [B_n]_\boldz^\boldZ \right) \wedge \neg \neg p(\boldv) \to p(\boldv) \mid \boldz \in |\intfolone|^{\sortsuper}, \boldv \in |\intfolone|^{\sortsuper}, \boldv \in [\boldt_\boldz^\boldZ]\}^\wedge
        \intertext{\hfill (Lemmas~\ref{lem:term.replacement}, \ref{lem:remove.grounding})}
        &\equiv_s \{ \tau \left( [B_1]_\boldz^\boldZ \right) \wedge \dots \wedge \tau\left( [B_n]_\boldz^\boldZ \right)  \to p(\boldv) \vee \neg p(\boldv) \mid \boldz \in |\intfolone|^{\sortsuper}, \boldv \in |\intfolone|^{\sortsuper}, \boldv \in [\boldt_\boldz^\boldZ]\}^\wedge.
        \intertext{\hfill (Lemma~\ref{lem:rule.head.choice.tau.helper})}
    \end{align*}
    Since $R$ is a choice rule with atom $\{p(\boldt)\}$ in the head, the rule instantiation process for~$R$ produces the following:
    \begin{align*}
        \tau(R) 
                &= \{\tau([B_1]_\boldz^\boldZ) \wedge \dots \wedge \tau([B_n]_\boldz^\boldZ) \to \bigwedge_{\boldr \in [\boldt_\boldz^\boldZ]} (p(\boldr) \vee \neg p(\boldr)) \mid \boldz \in |\intfolone|^{\sortsuper}\}^\wedge.
    \end{align*}
    Recall abbreviation $\bold{B}_\boldz^\boldZ$ from Case 1 and take any tuple of precomputed terms $\boldz \in |\intfolone|^{\sortsuper}$ of the same length as $\boldZ$. It remains to show that 
    $$
    \{ \bold{B}_\boldz^\boldZ \to (p(\boldv) \vee \neg p(\boldv))  \mid \boldv \in |\intfolone|^{\sortsuper}, \boldv \in [\boldt_\boldz^\boldZ]\}^\wedge \equiv_s 
       \bold{B}_\boldz^\boldZ \to \bigwedge_{\boldr \in [\boldt_\boldz^\boldZ]} (p(\boldr) \vee \neg p(\boldr)).
    $$
    Since $\bold{B}_\boldz^\boldZ$ is an infinitary formula and since $\boldt_\boldz^\boldZ$ is a tuple of ground terms of the same length as $\boldv$, this result follows directly from Lemma~\ref{lem:rule.head.tau.helper}.
    %
    %
    %
    %
    %
    \\
    \\
    \\\noindent \emph{Case 3}: $R$ is a constraint.
    By definition,
    \begin{align*}
        \grdstd\left(\tau^*(R)\right) 
        &= \grdstd\left( \forall \boldZ \big(\tau^B_\boldZ(B_1) \wedge \dots \wedge \tau^B_\boldZ(B_n) \to \bot\big) \right) \\
        &= \{ \grdstd\left( \tau^*_{\boldZ}(B_1)_\boldz^\boldZ \wedge \dots \wedge \tau^*_{\boldZ}(B_n)_\boldz^\boldZ \to \bot \right) \mid \boldz \in |\intfolone|^{\sortsuper}\}^\wedge \\
        &= \{ \grdstd\left( \tau^*_{\boldZ}([B_1]_\boldz^\boldZ) \right) \wedge \dots \wedge \grdstd\left(\tau^*_{\boldZ}([B_n]_\boldz^\boldZ) \right) \to \bot \mid \boldz \in |\intfolone|^{\sortsuper}\}^\wedge
    \end{align*}
    From Proposition~\ref{prop:cl} and Fact~\ref{fact:5:harrison}, it follows that
    \begin{align*}
        \grdstd\left(\tau^*(R)\right) 
        &\equiv_s \{ \tau([B_1]_\boldz^\boldZ)  \wedge \dots \wedge \tau([B_n]_\boldz^\boldZ) \to \bot \mid \boldz \in |\intfolone|^{\sortsuper}\}^\wedge
    \end{align*}
    Since $R$ is a constraint, the rule instantiation process for $R$ produces the following:
    \begin{align*}
        \tau(R) 
        &= \{ \neg \tau \left( [B_1]_\boldz^\boldZ \wedge \dots \wedge [B_n]_\boldz^\boldZ \right) \mid \boldz \in |\intfolone|^{\sortsuper}\}^\wedge \\
        &= \{ \tau \left( [B_1]_\boldz^\boldZ \wedge \dots \wedge [B_n]_\boldz^\boldZ \right) \to \bot \mid \boldz \in |\intfolone|^{\sortsuper}\}^\wedge \\
        &= \{ \tau \left( [B_1]_\boldz^\boldZ \right) \wedge \dots \wedge \tau\left( [B_n]_\boldz^\boldZ \right) \to \bot \mid \boldz \in |\intfolone|^{\sortsuper}\}^\wedge
    \end{align*}
    Thus, $\grdstd\left(\tau^*(R)\right) \equiv_s \tau(R)$.
\end{proof}


\paragraph{Proof of Proposition~\ref{prop:program}}

\begin{proof}
    By definition, $\tau^*\Pi$ is the first-order theory containing $\tau^*R$ for every rule $R$ in~$\Pi$.
    Thus, 
    $$\grdstd(\tau^*\Pi) = \{ \grdstd(\tau^*R) \mid R \in \Pi\}^\wedge.$$
    By definition,
    $\tau\Pi = \{ \tau R \mid R \in \Pi \}$.
    Note that 
    $$\{ \tau R \mid R \in \Pi \} \equiv_s  \{ \tau R \mid R \in \Pi \}^\wedge $$
    as by the definition of ht-satisfaction these sets of infinitary formulas share the same {\htmodel}s.
    From Fact~\ref{fact:5:harrison} and Proposition~\ref{prop:rule}, it follows that
    $$
    \{ \tau R \mid R \in \Pi \}^\wedge \equiv_s \{ \grdstd(\tau^*R) \mid R \in \Pi\}^\wedge
    $$
    Thus, $\grdstd(\tau^*(\Pi)) \equiv_s \tau \Pi$.
\end{proof}